\documentclass[pra,aps,nopacs,onecolumn,twoside,superscriptaddress, showkeys]{revtex4}

\usepackage{amsmath,amsfonts,amssymb,caption,color,epsfig,graphics,graphicx,hyperref,latexsym,mathrsfs,revsymb,amsthm,url,verbatim,epstopdf,mathtools,enumerate,stmaryrd,enumitem,appendix,multirow,physics,listings}

\hypersetup{colorlinks,linkcolor={blue},citecolor={blue},urlcolor={red}}

\newtheorem{definition}{Definition}

\newtheorem{lemma}[definition]{Lemma}

\newtheorem{theorem}[definition]{Theorem}

\newtheorem{conjecture}[definition]{Conjecture}

\newtheorem{example}[definition]{Example}

\begin{document}	

\title{On the saturated cases of the distillability conjecture}

\author{\Large Saiqi Liu}\email[]{liu\_saiqi\_sy2324111@buaa.edu.cn}
\affiliation{LMIB(Beihang University), Ministry of education, and School of Mathematical Sciences, Beihang University, Beijing 100191, China}

\author{\Large Lin Chen}\email[]{linchen@buaa.edu.cn (corresponding author)}
\affiliation{LMIB(Beihang University), Ministry of education, and School of Mathematical Sciences, Beihang University, Beijing 100191, China}

\date{\today}

% \Large

\begin{abstract}
    The distillability conjecture for two-copy $4 \times 4$ Werner states has been an open problem in quantum information for years. We investigate the conditions under which the conjectured inequality $\sigma_1^2(X)+\sigma_2^2(X)\leq 1/2$ becomes an equality. For all known cases where the conjecture has been verified, we characterize the saturation conditions and show that equality forces the matrices $A$ and $B$ to be $2\times 2$ block-diagonal. In particular, several previously obtained partial results, including the cases of one normal matrix, unitary similarity between $B$ and $-A$ or $-A^\top$, and anti‑diagonal block structures, are reduced to this common block‑diagonal structure. We also employ a manifold optimization method, which provides numerical evidence that the $2\times 2$ block‑diagonal structure is essential for saturating the inequality. Furthermore, we prove that the identified saturation points are critical points of the objective function on the constraint manifold.
\end{abstract}


\maketitle

\section{Introduction}

The entanglement distillability conjecture is recognized as one of the five fundamental theoretical problems within quantum information \cite{PRXQuantum.3.010101}. Essentially, this conjecture explores whether all entangled states with a non-positive partial transpose (NPT) can be asymptotically transformed into pure states through local operations and classical communication (LOCC). It has been confirmed that this process is feasible for any two-qubit entanglement \cite{hhh97}, and thus any qubit-qudit NPT entanglement. Over the past decades, numerous mathematical techniques have been formulated to address the distillability problem, such as the reduction criterion and entanglement witness for activation \cite{hh1999,klc02}. Efforts have also been made to quantify the upper bound of distillable entanglement \cite{rains1999,rains2001}. Additionally, the partial transpose operation on bipartite entangled states has been employed in distillation processes \cite{dcl00}. Another approach to studying the problem involves utilizing multiple copies of the target states \cite{br03}. Beyond its theoretical importance, entanglement distillation has practical applications, notably in the generation of security keys \cite{dw2005}. The distillability problem has also been investigated using methods like semidefinite programming \cite{vd06}. A crucial discovery is the intimate relationship between a state’s distillability and its rank \cite{cd11jpa}. Furthermore, it has been proven that bipartite NPT states with rank at most four are indeed 1-distillable \cite{cc08,cd16pra}. In particular, the distillability conjecture of two-copy $4 \times 4$ Werner states has raised a lot of attention recently \cite{5508622}. 

In this paper, we investigate the case where the inequality of the distillability conjecture is saturated. To date, significant progress has been made in the study of the conjecture. It has been verified in cases where one or both of the $4\times4$ matrices in the conjecture are normal matrices \cite{5508622,QIAN2021139}, as well as when both are $2 \times 2$ block diagonal matrices \cite{LIU202507}. Furthermore, the conjecture holds when one of the matrices is unitarily similar to either the other matrix or the transpose of the other \cite{LIU202507}. Additionally, the case where $A$ and $B$ are unitarily similar to matrices with two $2 \times 2$ off-diagonal blocks \cite{SIO2025152}, the case where $A$ or $B$ is unitarily similar to a direct sum of $2 \times 2$ trace zero matrices \cite{SIO2025152}, and the case where $A$ and $B$ are both monomial matrices \cite{Chen2021} have been proved to be valid. These established cases provide insights and a foundation for future research toward a complete resolution of the conjecture.

The core contribution is Theorem \ref{theorem: main result}, which outlines seven distinct cases where the distillability conjecture holds true. These scenarios are rigorously supported by a series of underlying lemmas: Lemma \ref{lemma:B=0} and Lemma \ref{lemma:A=0} deal with the cases where either $A$ or $B$ is zero; Lemma \ref{lemma:diag} deals with the case where $A$ and $B$ are diagonal matrices; and Lemma \ref{lemma:2 by 2 block-diagonal} deals with the case where $A$ and $B$ are $2 \times 2$ block-diagonal matrices. Furthermore, Lemmas \ref{lemma:one normal}, \ref{lemma:B and transpose of -A equivalent}, \ref{lemma:B and -A equivalent}, and \ref{lemma:2 2 anti} demonstrate that other cases such as those involving only one normal matrix, unitary similarity between $B$ and $-A^\top$ or $-A$, or matrices whose two $2 \times 2$ diagonal blocks are zero, can all be reduced to the block-diagonal conditions established in Lemma \ref{lemma:2 by 2 block-diagonal}. Finally, we utilize Lemmas \ref{lemma: smooth svd} through \ref{lemma: dense subset GL} to justify the smoothness of singular values and the use of dense subsets for our manifold optimization approach. It provides numerical evidence that implies the importance of the case where $A$ and $B$ are $2 \times 2$ block-diagonal matrices.

\section{Preliminaries}
\label{sec:pre}

Let $M_k$ denote the set of $k \times k$ matrices. We begin with the description of the distillability conjecture of two-copy $4 \times 4$ Werner states.
\begin{conjecture}\label{cj:distillability conjecture}
Let the matrices $A, B \in M_4(\mathbb{C})$ satisfy
	\begin{align}
		\operatorname{Tr}(A)=\operatorname{Tr}(B)=0, \label{eq:traceless condition}\\
		\|A\|_F^2+\|B\|_F^2=\frac{1}{4} \label{eq:frob norm condition}.
	\end{align}
We define the matrix $X \in M_{16}(\mathbb{C})$, 
	\begin{align}
		X=A \otimes I+I \otimes B.
	\end{align}
	Then we have
	\begin{align}
		\sigma_1(X)^2+\sigma_2(X)^2 \leq \frac{1}{2}, \label{ineq:main conjecture}
	\end{align}
	where $\sigma_1(X)$ and $\sigma_2(X)$ are respectively the largest and second-largest singular values of $X$.
\end{conjecture}

To date, there are several cases that have been solved. Here we list a few of them.
\begin{lemma}
	The distillability conjecture is valid for the following cases:
	\begin{enumerate}
        \item $A$ and $B$ are normal. \cite{5508622}
        \item $A$ or $B$ is normal. \cite{QIAN2021139}
		\item $A$ and $B$ are $2 \times 2$ block-diagonal. \cite{LIU202507}
        \item $B$ is unitarily similar to $-A$ or $-A^\top$. \cite{LIU202507}
        \item Two $2 \times 2$ diagonal blocks of both $A$ and $B$ are zero. \cite{SIO2025152}
	\end{enumerate}
\end{lemma}

\begin{definition}
    The Ky-Fan $2-2$ norm is defined as follows: $\| \cdot \|_{(2),2}: X \mapsto \sigma_1^2(X) + \sigma_2^2(X)$.
\end{definition}

\begin{lemma}\label{lemma: saturation of BW inequality with M_2 Y transpose}
    \cite{CHENG20132793} Let $X, Y \in M_2(\mathbb C)$. Then
    \begin{align}
        \left\|X Y-Y X^\top\right\|_F \leq \sqrt{2}\|X\|_F\|Y\|_F .
    \end{align}
    Moreover, equality holds if and only if there exists a unitary $U \in M_2$ such that
    \begin{align}
        U^* X U= \begin{bmatrix}
            x_{11} & x_{12} \\
            0 & -x_{11}
        \end{bmatrix} \quad 
        U^* Y \bar{U}=\begin{bmatrix}
            0 & y_{12} \\
            y_{21} & y_{22}
        \end{bmatrix}
    \end{align}
    where $\bar{x}_{12}\left(y_{12}+y_{21}\right)=2 \bar{x}_{11} y_{22}$.
\end{lemma}

\begin{lemma}\label{lemma: saturation of BW inequality with Y transpose}
    \cite{CHENG20132793} Let $n \geq 2$ and $X, Y \in M_n$. Then
    \begin{align}
        \left\|X Y-Y X^\top\right\|_F \leq \sqrt{2}\|X\|_{(2), 2}\|Y\|_F .
    \end{align}
    Moreover, equality holds if and only if one of the following conditions holds.
    \begin{enumerate}
        \item There exists a unitary $U \in M_n$ such that
        \begin{align}
            U^* X U=X_{11} \oplus X_{22} \quad \text { and } \quad U^* Y \bar{U}=Y_{11} \oplus 0,
        \end{align}
        where $X_{11}, Y_{11} \in M_2$ satisfy the equality condition in Lemma \ref{lemma: saturation of BW inequality with M_2 Y transpose}.
        \item There exists a unitary $U \in M_n$ such that
        \begin{align}
            U^* X U=X_{11} \oplus X_{22} \quad \text { and } \quad U^* Y \bar{U}=Y_{11} \oplus 0,
        \end{align}
        where $X_{11}, Y_{11} \in M_r(r \geq 3), X_{11}=\sigma_1(X) W$ where $W$ is unitary, and
        \begin{align}
            X_{11} Y_{11}+Y_{11} X_{11}^\top=0 .
        \end{align}
    \end{enumerate}
\end{lemma}

\begin{lemma}\label{lemma: neighborhood and saturation}
    If there exists $(\tilde A, \tilde B)$ such that $\sigma_1^2(\tilde X) + \sigma_2^2(\tilde X) > \frac{1}{2}$, then there exists $(\tilde{A}^\prime, \tilde{B}^\prime)$ such that $\sigma_1^2(\tilde{X}^\prime) + \sigma_2^2(\tilde{X}^\prime) = \frac{1}{2} + \epsilon$, where $\epsilon$ is arbitrarily small, and $(\tilde{A}^\prime, \tilde{B}^\prime)$ must be situated in the neighborhood of some $(\hat A, \hat B)$ where $\sigma_1^2(\hat X) + \sigma_2^2(\hat X) = \frac{1}{2}$.
\end{lemma}

\begin{proof}
    Since the singular values depend continuously on the matrix entries, the map $f$ is continuous. As there exist pairs achieving $1/2$, the existence of a small neighborhood exceeding $1/2$ follows trivially from continuity.
\end{proof}

According to Lemma \ref{lemma: neighborhood and saturation}, and considering the fact that we can treat Conjecture \ref{cj:distillability conjecture} as an optimization problem, cases of saturation are of great importance, in the following sections, we analyze these cases of saturation from both analytical and numerical perspectives.

\section{Main Result}\label{sec:main result}

The following theorem is the main result of our paper. It summarizes some saturated cases of the distillability conjecture, i.e. the case where $\sigma_1^2(X)+\sigma_2^2(X)=\frac{1}{2}$. We shall prove the items of this theorem by a sequence of lemmas.

\begin{theorem} \label{theorem: main result}
    We present seven cases of saturation for inequality \eqref{ineq:main conjecture}.
    \begin{enumerate}
		\item Let $X=A \otimes I_4$ in Conjecture \ref{cj:distillability conjecture}. Inequality \eqref{ineq:main conjecture} is saturated if and only if $\operatorname{rank} A=1$.
		\item Let $X=I_4 \otimes B$ in Conjecture \ref{cj:distillability conjecture}. Inequality \eqref{ineq:main conjecture} is saturated if and only if $\operatorname{rank} B=1$.
		\item Let $X=A \otimes I_4+I_4 \otimes B$ with $A=\operatorname{diag}\left(a_1, a_2, a_3, a_4\right), B=\operatorname{diag}\left(b_1, b_2, b_3, b_4\right)$ in Conjecture \ref{cj:distillability conjecture}. Inequality \eqref{ineq:main conjecture} is saturated if and only if $A$ and $B$ satisfy one of the following two statements (the case where $A$ or $B$ is zero is excluded),
		\begin{enumerate}
			\item $A=B=e^{i \phi} \operatorname{diag}\left(\frac{1}{4},-\frac{1}{4}, 0,0\right), \phi \in \mathbb{R}$.
			\item $A=e^{i \phi} \operatorname{diag}\left(-\frac{3}{8}, \frac{1}{8}, \frac{1}{8}, \frac{1}{8}\right), B=e^{i \phi} \operatorname{diag}\left(-\frac{1}{8},-\frac{1}{8}, \frac{1}{8}, \frac{1}{8}\right), \phi \in \mathbb{R}$.
		\end{enumerate}
		\item Let $A$ and $B$ be $2 \times 2$ block-diagonal matrices (the case where $A$ or $B$ are zero is excluded). Inequality \eqref{ineq:main conjecture} is saturated if and only if one of the following conditions is satisfied, 
		\begin{enumerate}
			\item $A=\begin{bmatrix}0 & a_{12} & 0 & 0 \\ a_{21} & 0 & 0 & 0 \\ 0 & 0 & 0 & 0 \\ 0 & 0 & 0 & 0\end{bmatrix}, B=\begin{bmatrix}0 & b_{12} & 0 & 0 \\ b_{21} & 0 & 0 & 0 \\ 0 & 0 & 0 & 0 \\ 0 & 0 & 0 & 0\end{bmatrix}$, and $a_{12} a_{21}=b_{12} b_{21}$.
			\item $A=\begin{bmatrix}a & 0 & 0 & 0 \\ 0 & a & 0 & 0 \\ 0 & 0 & -a & 0 \\ 0 & 0 & 0 & -a\end{bmatrix}, B=\begin{bmatrix}a & b_{12} & 0 & 0 \\ b_{21} & a & 0 & 0 \\ 0 & 0 & -a & 0 \\ 0 & 0 & 0 & -a\end{bmatrix}$, and $b_{12} b_{21}=4 a^2$.
		\end{enumerate} \label{case: 2 by 2 bloc-diagonal} 
		\item The saturated case where either $A$ or $B$ is normal is reduced to Case \ref{case: 2 by 2 bloc-diagonal}.
        \item The saturated case where $B$ is unitarily similar to $-A$ or $-A^\top$ is reduced to 
        Case \ref{case: 2 by 2 bloc-diagonal}.
        \item The saturated case where
        \begin{align}
            A = \begin{bmatrix}
                0 & A_1 \\
                A_2 & 0
            \end{bmatrix}, \quad
            B = \begin{bmatrix}
                0 & B_1 \\
                B_2 & 0
            \end{bmatrix},
        \end{align}
        and $A_1, A_2, B_1, B_2 \in M_2(\mathbb{C})$ is reduced to Case \ref{case: 2 by 2 bloc-diagonal}.
	\end{enumerate}
\end{theorem}

Part 1, 2, 3, 4, 5, and 7 of the theorem are respectively supported by Lemma \ref{lemma:B=0}, \ref{lemma:A=0}, \ref{lemma:diag}, \ref{lemma:2 by 2 block-diagonal},  \ref{lemma:one normal}, and \ref{lemma:2 2 anti}. Part 6 is supported by Lemma \ref{lemma:B and transpose of -A equivalent} and \ref{lemma:B and -A equivalent}.

\begin{lemma}\label{lemma:B=0}
	Let $X=A \otimes I_4$ in Conjecture \ref{cj:distillability conjecture}. Inequality $\sigma_1^2(X)+\sigma_2^2(X)=\frac{1}{2}$ is saturated if and only if $\operatorname{rank} A=1$.
\end{lemma}

\begin{proof}
	The proof is straightforward. We have
	\begin{align}
		\sigma_1^2(X)+\sigma_2^2(X)=2 \lambda_1\left(A A^{\dagger}\right) \leq 2 \operatorname{Tr} A A^{\dagger}=1 / 2 .
	\end{align}
	Hence, the inequality is saturated if and only if $\operatorname{rank} A=1$.
\end{proof}

\begin{lemma}\label{lemma:A=0}
	Let $X=I_4 \otimes B$ in Conjecture \ref{cj:distillability conjecture}. Inequality \eqref{ineq:main conjecture} is saturated if and only if $\operatorname{rank} B=1$.
\end{lemma}

\begin{proof}
	The proof is similar to the last lemma. 
\end{proof} 

\begin{lemma}\label{lemma:diag}
	Let $X=A \otimes I_4+I_4 \otimes B$ with $A=\operatorname{diag}\left(a_1, a_2, a_3, a_4\right), B=\operatorname{diag}\left(b_1, b_2, b_3, b_4\right)$ in Conjecture \ref{cj:distillability conjecture}. Inequality $\sigma_1^2(X)+\sigma_2^2(X)=\frac{1}{2}$ is saturated if and only if $A$ and $B$ satisfy one of the following two statements (the case where $A$ or $B$ is zero is excluded),
	\begin{enumerate}
		\item $A=B=e^{i \phi} \operatorname{diag}\left(\frac{1}{4},-\frac{1}{4}, 0,0\right), \phi \in \mathbb{R}$.
		\item $A=e^{i \phi} \operatorname{diag}\left(-\frac{3}{8}, \frac{1}{8}, \frac{1}{8}, \frac{1}{8}\right), B=e^{i \phi} \operatorname{diag}\left(-\frac{1}{8},-\frac{1}{8}, \frac{1}{8}, \frac{1}{8}\right), \phi \in \mathbb{R}$.
	\end{enumerate}
\end{lemma}
\begin{proof}
    We show the proof in Appendix \ref{sec:proof of lemma diag}.
\end{proof}

\begin{lemma}\label{lemma:2 by 2 block-diagonal}
	Let $A$ and $B$ be $2 \times 2$ block-diagonal matrices in Conjecture \ref{cj:distillability conjecture} (the case where $A$ or $B$ is zero is excluded). Then $\sigma_1^2(X)+\sigma_2^2(X)=\frac{1}{2}$ if and only if one of the following conditions is satisfied.
\begin{enumerate}
		\item[(i)] $A=\begin{bmatrix}
			0 & a_{12} & 0 & 0 \\ 
			a_{21} & 0 & 0 & 0 \\ 
			0 & 0 & 0 & 0 \\ 
			0 & 0 & 0 & 0
		\end{bmatrix}, B=\begin{bmatrix}
			0 & b_{12} & 0 & 0 \\ 
			b_{21} & 0 & 0 & 0 \\ 
			0 & 0 & 0 & 0 \\ 
			0 & 0 & 0 & 0
		\end{bmatrix}$, and $a_{12} a_{21}=b_{12} b_{21}$.
	\item[(ii)] $A=\begin{bmatrix}
			a & 0 & 0 & 0 \\ 
			0 & a & 0 & 0 \\ 
			0 & 0 & -a & 0 \\ 
			0 & 0 & 0 & -a
		\end{bmatrix}, B=\begin{bmatrix}
			a & b_{12} & 0 & 0 \\ 
			b_{21} & a & 0 & 0 \\ 
			0 & 0 & -a & 0 \\ 
			0 & 0 & 0 & -a
		\end{bmatrix}$, and $b_{12} b_{21}=4 a^2$.
	\end{enumerate}
\end{lemma} 

\begin{proof}
    We show the proof in Appendix \ref{appendix: proof 2 by 2 block-diagonal}.
\end{proof}

\begin{lemma}\label{lemma:one normal}
	The case where $A$ or $B$ is normal is reduced to Lemma \ref{lemma:2 by 2 block-diagonal}.
\end{lemma}

\begin{proof}
We review some conclusions given by \cite{QIAN2021139}.
	\begin{enumerate}
		\item Consider the case where $A$ and $B$ are real, $\sigma_1^2(X)$ and $\sigma_2^2(X)$ are from the same 4 by 4 block of $X$. According to equation (53c) of \cite{QIAN2021139}, $\sigma_1^2(X)+\sigma_2^2(X)=\frac{1}{2}$ only if $A$ and $B$ are both normal. 
		\item Consider the case where $A$ and $B$ are real, $\sigma_1^2(X)$ and $\sigma_2^2(X)$ are from different 4 by 4 blocks of $X$. According to Lemmas 18 and 19 of \cite{QIAN2021139}, $\sigma_1^2(X)+\sigma_2^2(X)=\frac{1}{2}$ only if $A$ is normal and $B$ is $2 \times 2$ block-diagonal.
		\item Consider the case where $A$ and $B$ are complex. According to equation (93) of \cite{QIAN2021139}, let $X=A \otimes I+ I \otimes B, X_1 \in M_{16}(\mathbb{R})$ and $X_2 \in M_{16}(\mathbb{R})$ be its real and imaginary parts, we have
		\begin{align}
			\sigma_1^2(X)+\sigma_2^2(X) \leq \sigma_1^2\left(X_1\right)+\sigma_2^2\left(X_1\right)+\sigma_1^2\left(X_2\right)+\sigma_2^2\left(X_2\right).
		\end{align}
	\end{enumerate}
    According to the three assertions above, we have $\sigma_1^2(X)+\sigma_2^2(X)=\frac{1}{2}$ only if $A$ is normal and $B$ is $2 \times 2$ block-diagonal (or vice versa). Thus, the case where $A$ or $B$ is normal is reduced to Lemma \ref{lemma:2 by 2 block-diagonal}.
\end{proof}

\begin{lemma}\label{lemma:B and transpose of -A equivalent}
    Let $-A^\top$ and $B$ be unitarily similar matrices in Conjecture \ref{cj:distillability conjecture}, then $\sigma_1^2(X) + \sigma_2^2(X) = \frac{1}{2}$ implies that they satisfy the conditions in Lemma \ref{lemma:2 by 2 block-diagonal}.
\end{lemma}

\begin{proof}
    We begin by briefly reviewing the proof of the validity of Conjecture \ref{cj:distillability conjecture} when $B$ and $-A^\top$ are unitarily similar. The detailed proof can be found in \cite{LIU202507}. Thanks to the variational characterization of singular values and vectorization of matrices, we can transform Conjecture \ref{cj:distillability conjecture} to the following form
    \begin{align}
        \sigma_1(X)^2 + \sigma_2(X)^2 = \max_{\substack{V_1, V_2 \in M_4(\mathbb{C}), \operatorname{Tr} V_1^* V_2 = 0 \\ \|V_1\|_F = \|V_2\|_F = 1}} \|B V_1 + V_1 A^\top\|_F^2 + \|B V_2 + V_2 A^\top\|_F^2.
    \end{align}
    Considering that $B$ and $-A^\top$ are unitarily similar, we replace $B$ with $-UA^\top U^*$
    \begin{align}
        \sigma_1(X)^2 + \sigma_2(X)^2 =& \max_{\substack{V_1, V_2 \in M_4(\mathbb{C}), \operatorname{Tr} V_1^* V_2 = 0 \\ \|V_1\|_F = \|V_2\|_F = 1}} \|-UA^\top U^* V_1 + V_1 A^\top\|_F^2 + \|-UA^\top U^* V_2 + V_2 A^\top\|_F^2 \\
        =& \max_{\substack{V_1, V_2 \in M_4(\mathbb{C}), \operatorname{Tr} V_1^* V_2 = 0 \\ \|V_1\|_F = \|V_2\|_F = 1}} \|-A^\top U^* V_1 + U^* V_1 A^\top\|_F^2 + \|-A^\top U^* V_2 + U^* V_2 A^\top\|_F^2. 
    \end{align}
    According to \cite{CHENG2010292}, we have $\| XY-YX \|_F \leq \sqrt{2} \|X\|_F \|Y\|_F$. Equality holds if and only if 
    \begin{enumerate}
        \item $X$ and $Y$ are simultaneously unitarily similar to matrices in $M_2 \oplus 0$.
        \item $\operatorname{Tr}X = \operatorname{Tr}Y = 0$.
        \item $\operatorname{Tr}(X^* Y) = 0$.
    \end{enumerate}
    From these conditions the conclusion follows.
\end{proof}

\begin{lemma}\label{lemma:B and -A equivalent}
    Let $-A$ and $B$ be unitarily similar matrices in Conjecture \ref{cj:distillability conjecture}. Then $\sigma_1^2(X) + \sigma_2^2(X) = \frac{1}{2}$ implies that they satisfy the conditions in Lemma \ref{lemma:2 by 2 block-diagonal}.
\end{lemma}

\begin{proof}
    Similar to the proof of previous lemma, we have
    \begin{align}
        \sigma_1(X)^2 + \sigma_2(X)^2 =& \max_{\substack{V_1, V_2 \in M_4(\mathbb{C}), \operatorname{Tr} V_1^* V_2 = 0 \\ \|V_1\|_F = \|V_2\|_F = 1}} \|-A U^* V_1 + U^* V_1 A^\top\|_F^2 + \|-A U^* V_2 + U^* V_2 A^\top\|_F^2. 
    \end{align}
    According to Lemma \ref{lemma: saturation of BW inequality with Y transpose}, we have $\sigma_1^2(X) + \sigma_2^2(X)=\frac{1}{2}$ if and only if one of the following conditions holds:
    \begin{enumerate}
        \item There exists a unitary $W \in M_4$ such that
        \begin{align}
            W^* A W=E_{11} \oplus E_{22} \quad \text { and } \quad W^* U^* V_i \bar{W}=F_{11} \oplus 0 \quad i=1,2,
        \end{align}
        where $E_{11}, F_{11} \in M_2$ satisfy the equality condition in Lemma \ref{lemma: saturation of BW inequality with M_2 Y transpose}.
        \item There exists a unitary $W \in M_4$ such that
        \begin{align}
            W^* A W=E_{11} \oplus E_{22} \quad \text { and } \quad W^* U^* V_i\bar{W}=F_{11} \oplus 0 \quad i=1,2,
        \end{align}
        where $E_{11}, F_{11} \in M_r(r \geq 3), E_{11}=\sigma_1(A) Z$ where $Z$ is unitary, 
        \begin{align}
            E_{11} F_{11}+F_{11} E_{11}^\top=0,
        \end{align}
        and $\operatorname{rank} A \leq 2$.
    \end{enumerate}
    We can see that in the first case, $A$ is necessarily $2 \times 2$ block-diagonal, and so is $B$. In the second case, we assume that $\sigma_1(A)>0$ (the case where $A=0$ is excluded), $\operatorname{rank} E_{11} = \operatorname{rank} \sigma_1(A)Z \geq 3$, which contradicts the fact that $\operatorname{rank} A \leq 2$.
\end{proof}

\begin{lemma}\label{lemma:2 2 anti}
    Let $A$ and $B$ be the matrices in Conjecture \ref{cj:distillability conjecture} with the following form:
    \begin{align}
        A = \begin{bmatrix}
            0 & A_1 \\
            A_2 & 0
        \end{bmatrix}, \quad
        B = \begin{bmatrix}
            0 & B_1 \\
            B_2 & 0
        \end{bmatrix},
    \end{align}
    where $A_1, A_2, B_1, B_2 \in M_2(\mathbb{C})$. Then $\sigma_1^2(X) + \sigma_2^2(X) = \frac{1}{2}$ implies that $A$ and $B$ satisfy the condition in Lemma \ref{lemma:2 by 2 block-diagonal}.
\end{lemma}

\begin{proof}
    We show the proof in Appendix \ref{appendix:proof 2 2 anti}.
\end{proof}

We also discover that the saturation conditions of the distillability conjecture are deeply related to orbit type stratification. Here, we introduce some definitions and lemmas from \cite{MICHEL200111}, and validate our findings in Lemma \ref{lemma:diag}.

\begin{definition}[Stabilizer]
    Let a group $\mathcal G$ act on a manifold $\mathcal M$. For a given point $m \in \mathcal M$, the stabilizer $\mathcal{G}_m$ of $m$ is defined as the subgroup consisting of all elements in $\mathcal G$ that leave $m$ fixed, that is:
    \begin{align}
        \mathcal{G}_{m} = \{g \in \mathcal G \mid g \cdot m = m\}.
    \end{align}
    If two points on an orbit satisfy the relation $m' = g \cdot m$, then their stabilizers satisfy the conjugation relation $\mathcal{G}_{m'} = g\mathcal{G}_{m}g^{-1}$.
\end{definition}

\begin{definition}[Conjugate Subgroups]
    Let $H$ and $H'$ be two subgroups of a group $\mathcal{G}$. If there exists an element $g \in \mathcal{G}$ such that $H' = gHg^{-1}$, then the subgroups $H$ and $H'$ are said to be conjugate.
\end{definition}

\begin{definition}[Stratum]
    In the action of a group $\mathcal{G}$ on a space $\mathcal M$, if two orbits have the same conjugacy class of stabilizers, they are said to belong to the same orbit type. A stratum is defined as the union of all orbits with the same orbit type. Equivalently, two points in the space belong to the same stratum if and only if their stabilizers are conjugate.
\end{definition}

\begin{definition}[Maximal Symmetry]
    The set of all strata can be endowed with a natural partial order by the conjugacy classes of their stabilizers. For two strata $S_1$ and $S_2$, if the stabilizer of $S_1$ (up to conjugation) is a proper subgroup of the stabilizer of $S_2$, then the local symmetry of $S_1$ is said to be strictly less than that of $S_2$ (denoted as $S_1 < S_2$).
    
    If the symmetry of a stratum $S_{\max}$ is not strictly less than any other stratum under the above partial order, then this stratum is said to have maximal symmetry.
\end{definition}

\begin{lemma}\label{lemma: maximal symmetry stratum}
    In the smooth action of a compact group $\mathcal G$ on a finite-dimensional compact manifold $\mathcal M$, every $\mathcal G$-invariant smooth function has orbits of extrema on every stratum with maximal symmetry.
\end{lemma}

According to Lemma \ref{lemma: maximal symmetry stratum}, it can be seen that within the manifold $\{(A, B) \mid \norm{A}_F^2+\norm{B}_F^2 = \frac{1}{4}, \operatorname{Tr} A = \operatorname{Tr} B = 0, A,B \in M_4(\mathbb C)\}$, there exist four strata with maximal symmetry:
\begin{align}
    A = \operatorname{diag}(\lambda_1, \lambda_1, -\lambda_1, -\lambda_1)&, B = \operatorname{diag}(\lambda_1, \lambda_1, -\lambda_1, -\lambda_1) \\
    A = \operatorname{diag}(\lambda_1, \lambda_1, -\lambda_1, -\lambda_1)&, B = \operatorname{diag}(\lambda_2, -\frac{\lambda_2}{3}, -\frac{\lambda_2}{3}, -\frac{\lambda_2}{3}) \\
    A = \operatorname{diag}(\lambda_2, -\frac{\lambda_2}{3}, -\frac{\lambda_2}{3}, -\frac{\lambda_2}{3})&, B = \operatorname{diag}(\lambda_1, \lambda_1, -\lambda_1, -\lambda_1) \\
    A = \operatorname{diag}(\lambda_2, -\frac{\lambda_2}{3}, -\frac{\lambda_2}{3}, -\frac{\lambda_2}{3})&, B = \operatorname{diag}(\lambda_2, -\frac{\lambda_2}{3}, -\frac{\lambda_2}{3}, -\frac{\lambda_2}{3})
\end{align}
These exactly correspond to the equality conditions obtained when treating both $A$ and $B$ as normal matrices:
\begin{enumerate}
    \item $A=B=e^{i \phi} \operatorname{diag}\left(\frac{1}{4},-\frac{1}{4}, 0,0\right), \phi \in \mathbb{R}$.
    \item $A=e^{i \phi} \operatorname{diag}\left(-\frac{3}{8}, \frac{1}{8}, \frac{1}{8}, \frac{1}{8}\right), B=e^{i \phi} \operatorname{diag}\left(-\frac{1}{8},-\frac{1}{8}, \frac{1}{8}, \frac{1}{8}\right), \phi \in \mathbb{R}$.
\end{enumerate}

\begin{lemma}\label{lemma:critical point}
    The following two cases in Lemma \ref{lemma:2 by 2 block-diagonal} are critical points on $\mathcal{X} = \{(A, B) \in M_4(\mathbb{C}) \times M_4(\mathbb{C}) | \operatorname{Tr} A = \operatorname{Tr} B = 0, \norm{A}_F^2 + \norm{B}_F^2 = \frac{1}{4}\}$, 
\begin{enumerate}
        \item $A=\begin{bmatrix}0 & a_{12} & 0 & 0 \\ a_{21} & 0 & 0 & 0 \\ 0 & 0 & 0 & 0 \\ 0 & 0 & 0 & 0\end{bmatrix}, B=\begin{bmatrix}0 & b_{12} & 0 & 0 \\ b_{21} & 0 & 0 & 0 \\ 0 & 0 & 0 & 0 \\ 0 & 0 & 0 & 0\end{bmatrix}$, and $a_{12} a_{21}=b_{12} b_{21}$.
        \item $A=\begin{bmatrix}a & 0 & 0 & 0 \\ 0 & a & 0 & 0 \\ 0 & 0 & -a & 0 \\ 0 & 0 & 0 & -a\end{bmatrix}, B=\begin{bmatrix}a & b_{12} & 0 & 0 \\ b_{21} & a & 0 & 0 \\ 0 & 0 & -a & 0 \\ 0 & 0 & 0 & -a\end{bmatrix}$, and $b_{12} b_{21}=4 a^2$.
	\end{enumerate}
\end{lemma}

\begin{proof}
    We denote the subgroup $I_2 \oplus \mathcal{U}(2)$ of unitary group $\mathcal{U}(4)$ by $\mathcal{G}$, and $M_2(\mathbb{C}) \oplus \mathbb{C}I_2 \cap \mathcal{X}$ by $\Sigma_\mathcal{G}$, every element in $\Sigma_\mathcal{G}$ is invariant under similarity transformations by elements in $\mathcal{G}$. According to Lemma \ref{lemma: smooth svd} and \ref{lemma:smooth}, the map $f: (A, B) \mapsto \sigma_1^2(X) + \sigma_2^2(X)$ is smooth at the two cases in Lemma \ref{lemma:2 by 2 block-diagonal}. Since $f(A, B) = f(UAU^*, VBV^*)$ for any $U,V \in \mathcal{U}(4)$, $f$ is $\mathcal{G}$ invariant. By definition, $\mathcal{X}$, when equipped with Euclidean inner product, is a Riemannian manifold, and $\mathcal{G}$ is a group of isometries on $\mathcal{X}$. According to the theorem in section 2 of \cite{Palais1979}, we obtain that the two cases in Lemma \ref{lemma:2 by 2 block-diagonal} are not only critical points on $\Sigma_\mathcal{G}$, but also critical points on $\mathcal{X}$.
\end{proof}

Lemma \ref{lemma:critical point} partially explains why $2 \times 2$ block-diagonal form appears constantly in the numerical experiment in Section \ref{sec:num evidence}. 

\section{Numerical Evidence}\label{sec:num evidence}

In this section, we aim to obtain additional examples of matrices $A$ and $B$ that achieve equality in Conjecture \ref{cj:distillability conjecture} by framing the conjecture as an optimization problem.
Specifically, we maximize the sum of the squares of the two largest singular values of the matrix $X = A \otimes I + I \otimes B$, where $A, B \in M_{4}(\mathbb{R})$ are matrices whose only off-diagonal entries are the only non-zero free variables (since any traceless matrix is unitarily similar to a matrix with all diagonal entries equal to zero), and all non-zero entries are subject to the equality constraint $\|A\|_F^2 + \|B\|_F^2 = \frac{1}{4}$.

We employ a manifold optimization approach. We flatten the 24 off-diagonal entries of $A$ and $B$ into a single vector $u \in \mathbb{R}^{24}$. Based on the constraint $\|u\|_F^2 = \frac{1}{4}$, we introduce the variable substitution $v = 2u$, such that $\|v\|_F^2 = 1$. Consequently, the original problem is equivalent to finding the optimal parameter vector $v$ on the 24-dimensional unit sphere manifold $\mathbb{S}^{23} = \{v \in \mathbb{R}^{24} \mid \|v\|_F = 1\}$. During each evaluation of the objective function, the point $v$ on the manifold is mapped back via the inverse transformation $u = \frac{v}{2}$ to populate the off-diagonal entries of matrices $A$ and $B$.

We define the objective function as $f(v) = \sigma_1^2(X(v)) + \sigma_2^2(X(v))$. Although singular values themselves are not smooth functions of the matrix entries, we have the following lemma:
\begin{lemma}\cite{BARBARINO2025465}\label{lemma: smooth svd}
    For an $M \times N$ matrix $X(t)$ that is smooth on a certain interval of $\mathbb{R}$, there exists the following decomposition:
    \begin{align}
        X(t)=U(t) \Sigma(t) V(t)^*
    \end{align}
    where $U(t)$ and $V(t)$ are smooth unitary matrices, and $\Sigma(t)$ is a smooth, diagonal, real-valued matrix.
\end{lemma}
The entries of the matrix $\Sigma(t)$ in the above lemma may be non-negative, and their absolute values are equal to the singular values of $X(t)$. Our objective function, $\sigma_1^2(X) + \sigma_2^2(X)$ being the sum of the squares of the singular values is independent of their signs; it is smooth along any curve at any point (with the exception of points where $\sigma_2(X) = \sigma_3(X)$; however, since the set of such points has measure zero within the entire domain, we may disregard them for the purposes of numerical experiments). Furthermore, regarding the relationship between smoothness along arbitrary curves and smoothness within a neighborhood of a point, we have the following fact.
\begin{lemma}\cite{Boman_1967}\label{lemma:smooth}
    Let $f: \mathbb{R}^d \to \mathbb{R}$. If, for any smooth curve $u \in C^{\infty}\left(\mathbb{R}, \mathbb{R}^d\right)$, the composite function $f \circ u$ is a smooth function on $\mathbb{R}$, then $f$ is a smooth function on $\mathbb{R}^d$.
\end{lemma}
Consequently, since $f(v)$ is locally smooth almost everywhere, it is possible to use an optimization method based on the Riemannian gradient. The code is provided in Appendix \ref{appendix: numerical code}.

\begin{example}\label{example: saturation}
	Here, we present five pairs of matrices $A$ and $B$ for which $\sigma_1^2(X) + \sigma_2^2(X) \geq 0.499999$ (retaining only four significant figures) below.
	\begin{enumerate}
		\item The first pair
		\begin{itemize}
			\item Matrices $A$ and $B$:
			\begin{align}
				A &= \begin{bmatrix}
					0 & 0.0002 & -0 & -0.1239 \\
					0.0001 & 0 & 0.0003 & -0.0046 \\
					-0 & 0.0006 & 0 & -0.3117 \\
					-0.0572 & -0.0021 & -0.1439 & 0 
				\end{bmatrix}, \\
				B &= \begin{bmatrix}
					0 & -0 & -0.1379 & -0.0948 \\
					-0 & 0 & 0.0632 & 0.0435 \\
					-0.2113 & 0.0969 & 0 & 0 \\
					-0.1453 & 0.0667 & 0 & 0 
				\end{bmatrix}.
			\end{align}
			\item The 16-dimensional vector formed by the singular values of matrix $X$:\\
			$[0.5, 0.5, 0.3355, 0.3355, 0.2821, 0.2821, 0.1841, 0.1841, 0.1548, 0.1548, 0, \cdots, 0]$.
			\item The difference between the vector of singular values of $A$ (sorted in descending order) and the vector of the absolute values of its eigenvalues (sorted in descending order):
			$[0.1076, -0.0731, 0, -0]$.
			\item The difference between the vector of singular values of $B$ (sorted in descending order) and the vector of the absolute values of its eigenvalues (sorted in descending order):
			$[0.0542, -0.0438, 0, -0]$.
		\end{itemize}
		
		\item The second pair
		\begin{itemize}
			\item Matrices $A$ and $B$:
			\begin{align}
				A &= \begin{bmatrix}
					0 & -0 & -0.0932 & 0 \\
					0 & 0 & 0.1866 & 0 \\
					0.0006 & -0.0012 & 0 & 0.0024 \\
					-0 & -0 & -0.3704 & 0 
				\end{bmatrix}, \\
				B &= \begin{bmatrix}
					0 & -0.002 & -0 & 0 \\
					0.1181 & 0 & 0.1335 & -0.1937 \\
					-0 & -0.0022 & 0 & 0 \\
					0 & 0.0032 & 0 & 0 
				\end{bmatrix}.
			\end{align}
			\item The 16-dimensional vector formed by the singular values of matrix $X$:\\
			$[0.5, 0.5, 0.4251, 0.4251, 0.2632, 0.2632, 0.0044, 0.0044, 0.0027, 0.0027, 0, \cdots, 0]$.
			\item The difference between the vector of singular values of $A$ (sorted in descending order) and the vector of the absolute values of its eigenvalues (sorted in descending order):
			$[0.391, -0.0313, 0, -0]$.
			\item The difference between the vector of singular values of $B$ (sorted in descending order) and the vector of the absolute values of its eigenvalues (sorted in descending order):
			$[0.2292, -0.0296, 0, -0]$.
		\end{itemize}
		
		\item The third pair
		\begin{itemize}
			\item Matrices $A$ and $B$:
			\begin{align}
				A &= \begin{bmatrix}
					0 & -0 & 0.3911 & -0.0554 \\
					0 & 0 & 0.0599 & -0.0085 \\
					-0.0982 & -0.0151 & 0 & 0 \\
					0.0139 & 0.0021 & -0 & 0 
				\end{bmatrix}, \\
				B &= \begin{bmatrix}
					0 & -0.0729 & 0.1262 & 0.0861 \\
					0.0729 & 0 & -0.0004 & 0.0533 \\
					-0.1262 & 0.0004 & 0 & -0.0928 \\
					-0.0861 & -0.0533 & 0.0928 & 0 
				\end{bmatrix}.
			\end{align}
			\item The 16-dimensional vector formed by the singular values of matrix $X$:\\
			$[0.5, 0.5, 0.3996, 0.3996, 0.2003, 0.2003, 0.2003, 0.2003, 0.1004, 0.1004, 0, \cdots, 0]$.
			\item The difference between the vector of singular values of $A$ (sorted in descending order) and the vector of the absolute values of its eigenvalues (sorted in descending order):
			$[0.1994, -0.0999, 0, -0]$.
			\item The difference between the vector of singular values of $B$ (sorted in descending order) and the vector of the absolute values of its eigenvalues (sorted in descending order):
			$[0, -0, 0, -0]$.
		\end{itemize}
		
		\item The fourth pair
		\begin{itemize}
			\item Matrices $A$ and $B$:
			\begin{align}
				A &= \begin{bmatrix}
					0 & -0 & -0.2686 & 0 \\
					0 & 0 & -0.1585 & 0 \\
					-0.0851 & -0.0502 & 0 & -0.0561 \\
					0 & -0 & -0.177 & 0 
				\end{bmatrix}, \\
				B &= \begin{bmatrix}
					0 & -0.002 & 0 & -0.0215 \\
					-0.0045 & 0 & 0.0277 & 0 \\
					0 & 0.0125 & 0 & 0.1335 \\
					-0.0474 & 0 & 0.295 & 0 
				\end{bmatrix}.
			\end{align}
			\item The 16-dimensional vector formed by the singular values of matrix $X$:\\
			$[0.5, 0.5, 0.3586, 0.3586, 0.3001, 0.3001, 0.1358, 0.1358, 0.1137, 0.1137, 0, \cdots, 0]$.
			\item The difference between the vector of singular values of $A$ (sorted in descending order) and the vector of the absolute values of its eigenvalues (sorted in descending order):
			$[0.1567, -0.0882, 0, -0]$.
			\item The difference between the vector of singular values of $B$ (sorted in descending order) and the vector of the absolute values of its eigenvalues (sorted in descending order):
			$[0.0982, -0.0661, 0, -0]$.
		\end{itemize}
		
		\item The fifth pair
		\begin{itemize}
			\item Matrices $A$ and $B$:
			\begin{align}
				A &= \begin{bmatrix}
					0 & 0.1114 & 0.186 & 0.0059 \\
					-0.1114 & 0 & -0.0241 & 0.0517 \\
					-0.186 & 0.0241 & 0 & 0.0877 \\
					-0.0059 & -0.0517 & -0.0877 & 0 
				\end{bmatrix}, \\
				B &= \begin{bmatrix}
					0 & 0 & -0.2695 & 0.1174 \\
					-0 & 0 & 0.1093 & -0.0476 \\
					0.1553 & -0.063 & 0 & 0 \\
					-0.0676 & 0.0274 & -0 & 0 
				\end{bmatrix}.
			\end{align}
			\item The 16-dimensional vector formed by the singular values of matrix $X$:\\
			$[0.5, 0.5, 0.3172, 0.3172, 0.2408, 0.2408, 0.2408, 0.2408, 0.1828, 0.1828, 0, \cdots, 0]$.
			\item The difference between the vector of singular values of $A$ (sorted in descending order) and the vector of the absolute values of its eigenvalues (sorted in descending order):
			$[0, -0, 0, -0]$.
			\item The difference between the vector of singular values of $B$ (sorted in descending order) and the vector of the absolute values of its eigenvalues (sorted in descending order):
			$[0.0764, -0.058, 0, -0]$.
		\end{itemize}
	\end{enumerate}
\end{example}

In the above five examples for matrices $A$ and $B$, the last two entries of the difference vector formed by subtracting the vector of absolute values of eigenvalues (arranged in descending order) from the vector of singular values (also arranged in descending order) were close to zero. For matrices of this type, we have the following lemma:

\begin{lemma}\cite{Horn_Johnson_1991} \label{lemma: semi normal}
    Let $A \in M_n(\mathbb{C})$, with singular values $\sigma_1(A) \geq \cdots \geq \sigma_n(A) \geq 0$ and eigenvalues $\left\{\lambda_i(A)\right\}$ satisfying $\left|\lambda_1(A)\right| \geq \cdots \geq\left|\lambda_n(A)\right|$. $A$ is a normal matrix if and only if $\sigma_i(A)=\left|\lambda_i(A)\right|$ for all $i=1, \ldots, n$. More generally, if $\sigma_i(A)= \left|\lambda_i(A)\right|$ for $i=1, \ldots, k$, then $A$ is unitarily similar to $D \oplus B$, where $D= \operatorname{diag}\left(\lambda_1, \ldots, \lambda_k\right)$ and $B \in M_{n-k}(\mathbb{C})$.
\end{lemma}

\begin{lemma}\label{lemma: semi normal 0}
    Let $A \in M_n(\mathbb{C})$ be an invertible matrix whose singular values satisfy $\sigma_1(A) \geq \dots \geq \sigma_n(A) > 0$ and whose eigenvalues satisfy $|\lambda_1(A)| \geq \dots \geq |\lambda_n(A)| > 0$. If $\sigma_i(A) = |\lambda_i(A)|$ holds for all $i = k+1, \dots, n$, then $A$ is unitarily similar to $D \oplus B$, where $D = \operatorname{diag}(\lambda_{k+1}(A), \dots, \lambda_n(A))$ and $B \in M_k(\mathbb{C})$.
\end{lemma}

\begin{proof}
    The proof is established by applying Lemma \ref{lemma: semi normal} to the inverse of matrix $A$.
\end{proof}

Lemma \ref{lemma: semi normal 0} indicates that the invertible cases discussed in Example \ref{example: saturation} can all be unitarily diagonalized or unitarily block-diagonalized into $2 \times 2$ blocks. Furthermore, restricting our consideration to invertible matrices is a reasonable assumption, as the set of invertible matrices constitutes a dense subset. Specifically, we have the following lemma:
\begin{lemma}\label{lemma: dense subset GL}
Define the following sets:
\begin{equation}
\begin{aligned}
&\mathcal X = \{(A, B) \in M_4(\mathbb{C}) \times M_4(\mathbb{C}) | \operatorname{Tr} A = \operatorname{Tr} B = 0, \norm{A}_F^2 + \norm{B}_F^2 = \frac{1}{4}\}, \\
&\mathcal Y = \{(A, B) \in M_4(\mathbb{C}) \times M_4(\mathbb{C}) | \operatorname{Tr} A = \operatorname{Tr} B = 0, \norm{A}_F^2 + \norm{B}_F^2 = \frac{1}{4}, A, B \in GL_4(\mathbb{C})\}.
\end{aligned}
\end{equation}
The set $\mathcal{Y}$ is a dense subset of $\mathcal{X}$. If Conjecture \ref{cj:distillability conjecture} holds for $\mathcal{Y}$, then Conjecture \ref{cj:distillability conjecture} also holds for $\mathcal{X}$. \end{lemma}

\begin{proof}
    Let $(A, B) \in \mathcal{X}$. To prove that $\mathcal{Y}$ is dense in $\mathcal{X}$, we must construct a sequence $\{(A_n, B_n)\}_{n=1}^{\infty} \subset \mathcal{Y}$ such that $\lim_{n \to \infty} (A_n, B_n) = (A, B)$ under the norm $\norm{\cdot}: A,B \mapsto \sqrt{\norm{A}_F^2 + \norm{B}_F^2}$. 
    
    First, we construct the diagonal matrix $E = \operatorname{diag}(1, 1, -1, -1) \in M_4(\mathbb{C})$. It is easily seen that $\operatorname{Tr}(E) = 0$ and $\det(E) = 1 \neq 0$; thus, $E \in GL_4(\mathbb{C})$. Introducing a parameter $t \in \mathbb{R}$, we define the perturbed matrices:
    \begin{equation}
    \tilde{A}(t) = A + tE, \quad \tilde{B}(t) = B + tE.
    \end{equation}
    By the linearity of the trace operator and the fact that $(A, B) \in \mathcal{X}$, we have $\operatorname{Tr}(\tilde{A}(t)) = \operatorname{Tr}(A) + t\operatorname{Tr}(E) = 0$ for any $t \in \mathbb{R}$. Similarly, $\operatorname{Tr}(\tilde{B}(t)) = 0$. 
    
We consider the determinant polynomials $P_A(t) = \det(\tilde{A}(t))$ and $P_B(t) = \det(\tilde{B}(t))$; both are polynomials in $t$ of degree at most 4. For $t \neq 0$, consider the asymptotic behavior of $P_A(t)$:
    \begin{equation}
    \lim_{t \to \infty} \frac{P_A(t)}{t^4} = \lim_{t \to \infty} \det\left(\frac{1}{t}A + E\right) = \det(E) = 1.
    \end{equation}
    This indicates that the coefficient of the $t^4$ term in $P_A(t)$ is $1$; thus, $P_A(t)$ is not identically zero. Similarly, it can be shown that $P_B(t) \not\equiv 0$. 
    The number of real roots for both $P_A(t)$ and $P_B(t)$ is finite. Therefore, there must exist a strictly monotonically decreasing sequence of positive real numbers $\{t_n\}_{n=1}^{\infty}$ converging to $0$, such that for any $n \geq 1$, we have $P_A(t_n) \neq 0$ and $P_B(t_n) \neq 0$. 
    
    Let $\tilde{A}_n = \tilde{A}(t_n)$ and $\tilde{B}_n = \tilde{B}(t_n)$, it follows that $\tilde{A}_n, \tilde{B}_n \in GL_4(\mathbb{C})$. Furthermore, to satisfy the norm equality constraint within $\mathcal{Y}$:
    \begin{equation}
        s_n = \sqrt{\lVert \tilde{A}_n \rVert_F^2 + \lVert \tilde{B}_n \rVert_F^2}.
    \end{equation}
    Since $\tilde{A}_n, \tilde{B}_n \in GL_4(\mathbb{C})$, they must be non-zero matrix; consequently, $s_n > 0$. We now define the normalized sequences:
    \begin{equation}
        A_n = \frac{1}{2s_n} \tilde{A}_n, \quad B_n = \frac{1}{2s_n} \tilde{B}_n.
    \end{equation}
    Next, we verify that for any $n \ge 1$, $(A_n, B_n) \in \mathcal{Y}$:
    \begin{itemize}
        \item Trace-free property: $\operatorname{Tr}(A_n) = \frac{1}{2s_n}\operatorname{Tr}(\tilde{A}_n) = 0$; similarly, $\operatorname{Tr}(B_n) = 0$. 
        \item Invertibility: Since $s_n > 0$, we have $\det(A_n) = \left(\frac{1}{2s_n}\right)^4 \det(\tilde{A}_n) \neq 0$; thus, $A_n \in GL_4(\mathbb{C})$. Similarly, $B_n \in GL_4(\mathbb{C})$. 
        \item Norm constraint:
        \begin{equation}
            \lVert A_n \rVert_F^2 + \lVert B_n \rVert_F^2 = \frac{1}{4s_n^2} \left( \lVert \tilde{A}_n \rVert_F^2 + \lVert \tilde{B}_n \rVert_F^2 \right) = \frac{s_n^2}{4s_n^2} = \frac{1}{4}.
        \end{equation}
    \end{itemize}
    This demonstrates that the sequence $\{(A_n, B_n)\}_{n=1}^{\infty} \subset \mathcal{Y}$. 
    
    Finally, we compute the limit of this sequence:
    \begin{align}
        \lim_{n \to \infty} s_n &= \sqrt{\lVert A \rVert_F^2 + \lVert B \rVert_F^2} = \frac{1}{2}, \\
        \lim_{n \to \infty} A_n &= \lim_{n \to \infty} \frac{1}{2s_n} \tilde{A}_n = \frac{1}{2 \cdot (1/2)} A = A, \\
        \lim_{n \to \infty} B_n &= \lim_{n \to \infty} \frac{1}{2s_n} \tilde{B}_n = \frac{1}{2 \cdot (1/2)} B = B.
    \end{align}
    In summary, for any point $(A, B) \in \mathcal{X}$, there exists a sequence contained entirely within $\mathcal{Y}$ that converges to $(A, B)$; this completes the proof that $\mathcal{Y}$ is dense in $\mathcal{X}$. If there exists a counterexample to Conjecture \ref{cj:distillability conjecture} within $\mathcal{X}$ suppose, for instance, that $\sigma_1^2(\hat X) + \sigma_2^2(\hat X) = \frac{1}{2} + \epsilon$, then, since $\mathcal{Y}$ is dense in $\mathcal{X}$ and $\sigma_1^2(X) + \sigma_2^2(X)$ is a continuous function, there must exist a pair $(\tilde A, \tilde B) \in \mathcal{Y}$ such that $\sigma_1^2(\tilde X) + \sigma_2^2(\tilde X) = \frac{1 + \epsilon}{2}$. 
\end{proof}

\section{Conclusion}\label{sec:conclusion}

In this paper, we have investigated the saturation conditions of the distillability conjecture, identify seven distinct saturation scenarios. Crucially, our analysis reveals that more complex matrix structures, including those involving only one normal matrix, unitary similarity between the matrices, or matrices whose two $2 \times 2$ diagonal blocks are zero, can all be reduced to the $2 \times 2$ block-diagonal conditions established in Lemma \ref{lemma:2 by 2 block-diagonal}. To further substantiate our theoretical findings, we incorporate a manifold optimization approach. By justifying the smoothness of singular values and the applicability of dense subsets, we provide compelling numerical evidence that underscores the critical importance of the $2 \times 2$ block-diagonal structure in achieving equality within the distillability conjecture. These synthesized results not only clarify the boundary conditions of the distillability conjecture but also provide a potential pathway for future research toward its complete resolution.

\section*{ACKNOWLEDGMENTS}
Authors were supported by the NNSF of China (Grant No. 12471427). 

\section*{DISCLOSURE STATEMENT}
The authors report there are no competing interests to declare.

\appendix

\section{Proof of Lemma \ref{lemma:diag}}\label{sec:proof of lemma diag}

Let $A=\operatorname{diag}(a_1,a_2,a_3,a_4)$ and $B=\operatorname{diag}(b_1,b_2,b_3,b_4)$ with complex numbers $a_i,b_j\in \mathbb{C}$, such that $\sum_j a_j=\sum_k b_k=0$ and $\sum_j(\left| a_j \right|^2+\left| b_j \right|^2)=1/4$. Evidently $X$ is a diagonal matrix with diagonal elements $a_j+b_k$ with the subscript $j,k=1,2,3,4$. Hence the set of singular values $\sigma_j(X)$ is the same as the set of absolute values of $X$'s eigenvalues $\left| a_j+b_k \right|$. Up to the subscript permutation, we can assume that $\left| a_1+b_1 \right|$ is the largest singular value. We have two cases in terms of the second largest singular value, 
\begin{enumerate}
	\item $\sigma_2(X) = \left| a_2+b_2 \right|$, $\sigma_1(X)^2 + \sigma_2(X)^2 = \left| a_1+b_1 \right|^2+\left| a_2+b_2 \right|^2$. \label{case: normal 1}
	\item $\sigma_2(X) = \left| a_1+b_2 \right|$, $\sigma_1(X)^2 + \sigma_2(X)^2 = \left| a_1+b_1 \right|^2+\left| a_1+b_2 \right|^2$. \label{case: normal 2}
\end{enumerate}
In case \ref{case: normal 1}, we claim that $a_1 = b_1 = -a_2 = -b_2$, $\left| a_1 \right| = \left| b_1 \right| = \left| a_2 \right| = \left| b_2 \right| = \frac{1}{4}$ and $b_3 = b_4 = a_3 = a_4 = 0$. In fact, let us suppose that $\sigma_1(X)^2 + \sigma_2(X)^2 = \left| a_1+b_1 \right|^2+\left| a_2+b_2 \right|^2$ reaches the maximum $\frac{1}{2}$. By calculations, we have
\begin{align}
	\left| a_1 - b_1 \right|^2 + \left| a_2 - b_2 \right|^2 &= 2 \left( \left| a_1 \right|^2 + \left| b_1 \right|^2 + \left| a_2 \right|^2 + \left| b_2 \right|^2 \right) - \left(\left| a_1 + b_1 \right|^2 + \left| a_2 + b_2 \right|^2\right) \\
	&\leq 2 \times \frac{1}{4} - \frac{1}{2} = 0. 
\end{align}
Since $\left| a_1 - b_1 \right|^2 + \left| a_2 - b_2 \right|^2$ is nonnegative, we have $a_1 = b_1$ and $a_2 = b_2$. Consequently, we obtain 
\begin{align}
	\left| 2 a_1 \right|^2 + \left| 2 a_2 \right|^2 = \frac{1}{2}.
\end{align}
Since $\sum_{i=1}^4 \left| a_i \right|^2 + \left| b_i \right|^2 = \frac{1}{4}$, we have $b_3 = b_4 = a_3 = a_4 = 0$. According to \eqref{eq:traceless condition}, we obtain $a_1 = b_1 = -a_2 = -b_2$ and $\left| a_1 \right| = \left| b_1 \right| = \left| a_2 \right| = \left| b_2 \right| = \frac{1}{4}$. One can verify that $\left| a_1 + b_1 \right|$ and $\left| a_2 + b_2 \right|$ are two largest singular values of $X$. We have proven the claim \ref{case: normal 1}.
Next we investigate case \ref{case: normal 2}. We have
\begin{eqnarray}\label{eq:func}
	\left| a_1+b_1 \right|^2+\left| a_1+b_2 \right|^2=\frac{1}{2}(\left| b_1-b_2 \right|^2+\left| 2a_1+b_1+b_2 \right|^2).
\end{eqnarray}
The maximum of the expression above is reached when $a_1$ and $b_1+b_2$ have the same phase. Up to a global phase on the matrix $X$ in Conjecture \ref{cj:distillability conjecture}, we can assume that $a_1\ge0$ and $b_1+b_2\ge0$. Additionally, we can assume that $b_1 = e_1 + h i$, $b_2 = e_2 - h i$, $b_1 + b_2 = e_1 + e_2 = c_1$, and $\left| b_1 \right|^2 + \left| b_2 \right|^2 = e_1^2 + e_2^2 + 2 h^2 = c_2$. By calculations, we have
\begin{align}
	\left| b_1 - b_2 \right|^2 &= (e_1 - e_2)^2 + 4 h^2 = 4 h^2 + 2 e_1^2 + 2 e_2^2 - (e_1 + e_2)^2 \\
	&= 2 c_2 - c_1^2.
\end{align}
Then equation \eqref{eq:func} becomes
\begin{align}
	\frac{1}{2}\left(\left| b_1-b_2 \right|^2+\left| 2a_1+b_1+b_2 \right|^2\right) &= \frac{1}{2}\left(2 c_2 - c_1^2 + \left(2 a_1 + c_1\right)^2\right) \\
	&= 2 a_1^2 + 2 a_1 c_1 +  c_2. \label{eq: expre to be opt}
\end{align}
To determine under what conditions the expression \eqref{eq: expre to be opt} reaches its maximum $\frac{1}{2}$, we can use techniques from optimization to find out the necessary conditions for reaching the maximum. Here, we treat the expression \eqref{eq: expre to be opt} as a function $f$ with variables $c_1, c_2, a_1$ to be optimized
\begin{align}
	f(c_1, c_2, a_1) = 2 a_1^2 + 2 a_1 c_1 + c_2,
\end{align}
under the constraints defined as follows
\begin{align}
	&g_1(c_1, b_3, b_4) = c_1 + b_3 + b_4 = 0, \label{eq: constraint 1} \\
	&g_2(a_1, a_2, a_3, a_4) = a_1 + a_2 + a_3 + a_4 = 0, \label{eq: constraint 2} \\
	&g_3(a_1, a_2, a_3, a_4, c_2, b_3, b_4) = c_2 + \left| b_3 \right|^2 + \left| b_4 \right|^2 + a_1^2 + \left| a_2 \right|^2 + \left| a_3 \right|^2 + \left| a_4 \right|^2 - \frac{1}{4} = 0. \label{eq: constraint 3} 
\end{align}
Suppose that $c_1$ and $a_1$ are constants, which means the sum of $b_3$ and $b_4$ as well as the sum of $a_2, a_3$ and $a_4$ are fixed, we can let $b_3 = b_4$ and $a_2 = a_3 = a_4$ while minimizing $\left| b_3 \right|^2 + \left| b_4 \right|^2$ and $\left| a_2 \right|^2 + \left| a_3 \right|^2 + \left| a_4 \right|^2$ to make $c_2$ larger, thus increasing the value of $f$. Since $a_1$ and $c_1$ are real, based on the reasoning above, \eqref{eq: constraint 1} and \eqref{eq: constraint 2}, we can assume that $b_3 = b_4 = b \in \mathbb{R}$ and $a_2 = a_3 = a_4 = a \in \mathbb{R}$. According to constraints \eqref{eq: constraint 1}, \eqref{eq: constraint 2} and \eqref{eq: constraint 3}, we have
\begin{align}
	c_1 &= -2b, \\
	a_1 &= -3a, \\
	c_2 &= \frac{1}{4} - 12a^2 - 2b^2, \\
	f(c_1, c_2, a_1) &= 2 \times (-3a)^2 + 2 \times (-3a) \times (-2b) + \frac{1}{4} - 12a^2 - 2b^2 \\
	&= 12ab + 6a^2 - 2b^2 + \frac{1}{4}.
\end{align}
Recall that $c_1 = b_1 + b_2 = -2b$, so the minimum of $c_2$ is achieved when $b_1 = b_2$, thus $c_2 \geq 2b^2$. According to \eqref{eq: constraint 3}, we can reformulate the optimization problem as follows
\begin{align}
	\min -f(c_1, c_2, a_1) = -f(a, b) = -12ab - 6a^2 + 2b^2 - \frac{1}{4},
\end{align}
under the constraint
\begin{align}
	g(a, b) = c_2 - 2b^2 = \frac{1}{4} - 12a^2 - 4b^2 \geq 0. 
\end{align}
We define the Lagrangian $l:(a, b, \mu) \mapsto -12ab - 6a^2 + 2b^2 - \frac{1}{4} + \mu (\frac{1}{4} - 12a^2 - 4b^2)$, the two following conditions are necessary for reaching the minimum of $-f$
\begin{align}
	\frac{\partial l}{\partial a} &= -12b - 12a - 24 \mu a = 0, \\
	\frac{\partial l}{\partial b} &= -12a + 4b - 8 \mu b = 0.
\end{align}
By calculations, we have $b = a = 0$ or $\mu = 1$ or $-1$. In case $b = a = 0$, we have $A = 0$, it is already covered by Lemma \ref{lemma:B=0} and Lemma \ref{lemma:A=0}, we only need to examine the case where $\mu = 1$ or $-1$. When $\mu = 1$, we have $b = -3a$, $f(a, b) = -36a^2 + 6a^2 - 18a^2 + \frac{1}{4} \leq \frac{1}{4}$, thus this case is excluded. When $\mu = -1$, we have $b = a$, $f(a, b) = 16a^2 + \frac{1}{4}$. To attain the expected value $\frac{1}{2}$, $a = \frac{1}{8}$ or $-\frac{1}{8}$, we have $A = \operatorname{diag}\left( -\frac{3}{8}, \frac{1}{8}, \frac{1}{8}, \frac{1}{8}\right)$, $B = \operatorname{diag}\left(b_1, b_2, \frac{1}{8}, \frac{1}{8}\right)$ or $A = \operatorname{diag}\left( \frac{3}{8}, -\frac{1}{8}, -\frac{1}{8}, -\frac{1}{8}\right)$, $B = \operatorname{diag}\left(b_1, b_2, -\frac{1}{8}, -\frac{1}{8}\right)$. The case $a = \frac{1}{8}$ and the case $a = -\frac{1}{8}$ are similar, we will only work on the case $a = \frac{1}{8}$
\begin{align}
	b_1 + b_2 &= -\frac{1}{4}, \\
	b_1^2 + b_2^2 &= \frac{1}{32}. 
\end{align}
By calculations, we have $b_1 = b_2 = -\frac{1}{8}$. Recall that we perform a global phase change at first. Now, $A = e^{i\phi}\operatorname{diag}\left( -\frac{3}{8}, \frac{1}{8}, \frac{1}{8}, \frac{1}{8}\right)$, $B = e^{i\phi}\operatorname{diag}\left(-\frac{1}{8}, -\frac{1}{8}, \frac{1}{8}, \frac{1}{8}\right)$, $\phi \in \mathbb{R}$. The case $a = -\frac{1}{8}$ is already covered by the phase change $e^{i\phi}$. After verification, $\left| a_1 + b_1 \right|^2 + \left| a_1 + b_2 \right|^2 = \frac{1}{2}$, equations \eqref{eq:traceless condition} and \eqref{eq:frob norm condition} are indeed satisfied.

\section{Proof of Lemma \ref{lemma:2 by 2 block-diagonal}}\label{appendix: proof 2 by 2 block-diagonal}

We begin by reviewing how to prove the distillability conjecture.
\begin{align}
	A = 
	\begin{bmatrix}
		A_{11} & 0 \\
		0 & A_{22}
	\end{bmatrix}
	=
	\begin{bmatrix}
		a & a_{12} & 0 & 0 \\
		a_{21} & a & 0 & 0 \\
		0 & 0 & -a & a_{34} \\
		0 & 0 & a_{43} & -a
	\end{bmatrix},
	\quad
	B = 
	\begin{bmatrix}
		B_{11} & 0 \\
		0 & B_{22}
	\end{bmatrix}
	=
	\begin{bmatrix}
		b & b_{12} & 0 & 0 \\
		b_{21} & b & 0 & 0 \\
		0 & 0 & -b & b_{34} \\
		0 & 0 & b_{43} & -b
	\end{bmatrix}.
\end{align}
Then we calculate $X = A \otimes I + I \otimes B$, perform a permutation transformation and get $Y$
\begin{align}
	Y = \begin{bmatrix}
		Y_{11} & 0 & 0 & 0 \\
		0 & Y_{22} & 0 & 0 \\
		0 & 0 & Y_{33} & 0 \\
		0 & 0 & 0 & Y_{44}
	\end{bmatrix},
\end{align}
where
\begin{align}
	Y_{11} = 
	\begin{bmatrix}
		b+a & b_{12} & a_{12} & 0 \\
		b_{21} & b+a & 0 & a_{12} \\
		a_{21} & 0 & b+a & b_{12} \\
		0 & a_{21} & b_{21} & b+a
	\end{bmatrix}, 
	Y_{22} = 
	\begin{bmatrix}
		-b+a & b_{34} & a_{12} & 0 \\
		b_{43} & -b+a & 0 & a_{12} \\
		a_{21} & 0 & -b+a & b_{34} \\
		0 & a_{21} & b_{43} & -b+a
	\end{bmatrix}, \\
	Y_{33} = 
	\begin{bmatrix}
		b-a & b_{12} & a_{34} & 0 \\
		b_{21} & b-a & 0 & a_{34} \\
		a_{43} & 0 & b-a & b_{12} \\
		0 & a_{43} & b_{21} & b-a
	\end{bmatrix},
	Y_{44} = 
	\begin{bmatrix}
		-b-a & b_{34} & a_{34} & 0 \\
		b_{43} & -b-a & 0 & a_{34} \\
		a_{43} & 0 & -b-a & b_{34} \\
		0 & a_{43} & b_{43} & -b-a
	\end{bmatrix}.
\end{align}
We observe that 
\begin{align}
	Y_{11} = A_{11} \otimes I_2 + I_2 \otimes B_{11}, \\
	Y_{22} = A_{11} \otimes I_2 + I_2 \otimes B_{22}, \\
	Y_{33} = A_{22} \otimes I_2 + I_2 \otimes B_{11}, \\
	Y_{44} = A_{22} \otimes I_2 + I_2 \otimes B_{22}.
\end{align}
There are three different cases we have to consider
\begin{enumerate}[label = \roman*.]
	\item The largest two singular values are the two largest singular values of $Y_{11}$, \label{1st case: 2 2 bloc-diag}
	\item $\sigma_1(Y)$ is the largest singular value of $Y_{11}$, and $\sigma_2(Y)$ is the largest singular value of $Y_{22}$, \label{2nd case: 2 2 bloc-diag}
	\item $\sigma_1(Y)$ is the largest singular value of $Y_{11}$, and $\sigma_2(Y)$ is the largest singular value of $Y_{44}$. \label{3rd case: 2 2 bloc-diag}
\end{enumerate}

The saturation conditions for the three cases above are respectively demonstrated by Lemmas \ref{lemma: saturation of 1st case: 2 2 bloc-diag}, \ref{lemma: saturation of 2nd case: 2 2 bloc-diag} and \ref{lemma: saturation of 3rd case: 2 2 bloc-diag}.

\textbf{Initially, we consider under what conditions the inequality is saturated in case \ref{1st case: 2 2 bloc-diag}}

\begin{lemma}\label{lemma: saturation of 1st case: 2 2 bloc-diag}
    In case \ref{1st case: 2 2 bloc-diag}, $\sigma_1^2(X) + \sigma_2^2(X) = \frac{1}{2}$ implies that $A$ and $B$ satisfy one of the following conditions:
    \begin{enumerate}
    	\item $A = \begin{bmatrix}
    		0 & a_{12} & 0 & 0 \\
    		a_{21} & 0 & 0 & 0 \\
    		0 & 0 & 0 & 0 \\
    		0 & 0 & 0 & 0
    	\end{bmatrix}, 
    	B = \begin{bmatrix}
    		0 & b_{12} & 0 & 0 \\
    		b_{21} & 0 & 0 & 0 \\
    		0 & 0 & 0 & 0 \\
    		0 & 0 & 0 & 0
    	\end{bmatrix}$, and $a_{12}a_{21} = b_{12}b_{21}$. \label{case: first case of case i} \\
    	\item $A = \begin{bmatrix}
    		a & 0 & 0 & 0 \\
    		0 & a & 0 & 0 \\
    		0 & 0 & -a & 0 \\
    		0 & 0 & 0 & -a
    	\end{bmatrix}, 
    	B = \begin{bmatrix}
    		a & b_{12} & 0 & 0 \\
    		b_{21} & a & 0 & 0 \\
    		0 & 0 & -a & 0 \\
    		0 & 0 & 0 & -a
    	\end{bmatrix}$, and $b_{12}b_{21} = 4a^2$. \label{case: second case of case i}
    \end{enumerate}
\end{lemma}

\begin{proof}
    We consider
    \begin{align}
    	\sigma_1(Y)^2 + \sigma_2(Y)^2 &= \sigma_1(Y_{11})^2 + \sigma_2(Y_{11})^2 \leq \left\| Y_{11} \right\|_F^2 \label{ineq: 1st case 1} \\ 
    	&= 4\left| b+a \right|^2 + 2(\left| a_{12} \right|^2 + \left| a_{21} \right|^2 + \left| b_{12} \right|^2 + \left| b_{21} \right|^2) \\
    	&\leq 8 (\left| a \right|^2 + \left| b \right|^2) + 2(\left| a_{12} \right|^2 + \left| a_{21} \right|^2 + \left| b_{12} \right|^2 + \left| b_{21} \right|^2) \leq 2 (\left\| A \right\|_F^2 + \left\| B \right\|_F^2) = \frac{1}{2}. \label{ineq: 1st case 2}
    \end{align}
    The two inequalities in \eqref{ineq: 1st case 2} take equalities if and only if $a_{34} = a_{43} = b_{34} = b_{43} = \left| b-a \right| = 0 $. The inequality in \eqref{ineq: 1st case 1} takes equality if and only if $\operatorname{rank} Y_{11} \leq 2$, which is equivalent to all third-order minors of $Y_{11}$ being zero, we obtain 3 cases
    \begin{enumerate}
    	\item $a=0, a_{12}a_{21} = b_{12}b_{21}$.
    	\item $a_{12} = a_{21} = 0, b_{12}b_{21} = 4a^2$.
    	\item $b_{12} = b_{21} = 0, a_{12}a_{21} = 4a^2$.
    \end{enumerate}
    We can verify that all of the three cases above saturate the original conjecture. In fact, the second and the third one are equivalent by exchanging $A$ and $B$. Thus, up to unitary similarity, we obtain the following two cases while excluding the case we already discuss.
    \begin{enumerate}
    	\item $A = \begin{bmatrix}
    		0 & a_{12} & 0 & 0 \\
    		a_{21} & 0 & 0 & 0 \\
    		0 & 0 & 0 & 0 \\
    		0 & 0 & 0 & 0
    	\end{bmatrix}, 
    	B = \begin{bmatrix}
    		0 & b_{12} & 0 & 0 \\
    		b_{21} & 0 & 0 & 0 \\
    		0 & 0 & 0 & 0 \\
    		0 & 0 & 0 & 0
    	\end{bmatrix}$, and $a_{12}a_{21} = b_{12}b_{21}$. \\
    	\item $A = \begin{bmatrix}
    		a & 0 & 0 & 0 \\
    		0 & a & 0 & 0 \\
    		0 & 0 & -a & 0 \\
    		0 & 0 & 0 & -a
    	\end{bmatrix}, 
    	B = \begin{bmatrix}
    		a & b_{12} & 0 & 0 \\
    		b_{21} & a & 0 & 0 \\
    		0 & 0 & -a & 0 \\
    		0 & 0 & 0 & -a
    	\end{bmatrix}$, and $b_{12}b_{21} = 4a^2$.
    \end{enumerate}
\end{proof}

\textbf{Secondly, we consider under what condition the inequality is saturated in case \ref{2nd case: 2 2 bloc-diag}}

\begin{lemma}\label{lemma: saturation of 2nd case: 2 2 bloc-diag}
    In case \ref{2nd case: 2 2 bloc-diag}, the saturation condition is included in case \ref{1st case: 2 2 bloc-diag}.
\end{lemma}

\begin{proof}
    We first review how to prove case \ref{2nd case: 2 2 bloc-diag}, and find out some necessary conditions of saturation.

    We consider
    \begin{align}
    	Z_{11} = Y_{11} {Y_{11}}^* = 
    	\begin{bmatrix}
    		{Z_{11}}^{(11)} & {Z_{11}}^{(12)} \\
    		{{Z_{11}}^{(12)}}^* & {Z_{11}}^{(22)}
    	\end{bmatrix},
    	\quad
    	Z_{22} = Y_{22} {Y_{22}}^* = 
    	\begin{bmatrix}
    		{Z_{22}}^{(11)} & {Z_{22}}^{(12)} \\
    		{{Z_{22}}^{(12)}}^* & {Z_{22}}^{(22)}
    	\end{bmatrix},
    \end{align}
    where
    \begin{align}
    	{Z_{11}}^{(11)} = 
    	\begin{bmatrix}
    		|b + a|^2 + |b_{12}|^2 + |a_{12}|^2 & (b + a) \bar b_{21} + (\bar b + \bar a) b_{12} \\
    		(\bar b + \bar a) b_{21} + (b + a) \bar b_{12} & |b + a|^2 + |b_{21}|^2 + |a_{12}|^2
    	\end{bmatrix}, \\
    	{Z_{11}}^{(22)} = 
    	\begin{bmatrix}
    		|b + a|^2 + |b_{12}|^2 + |a_{21}|^2 & (b + a) \bar b_{21} + (\bar b + \bar a) b_{12} \\
    		(\bar b + \bar a) b_{21} + (b + a) \bar b_{12} & |b + a|^2 + |b_{21}|^2 + |a_{21}|^2
    	\end{bmatrix}, \\
    	{Z_{22}}^{(11)} = 
    	\begin{bmatrix}
    		|a - b|^2 + |b_{34}|^2 + |a_{12}|^2 & (a - b) \bar b_{43} + (\bar a - \bar b) b_{34} \\
    		(\bar a - \bar b) b_{43} + (a - b) \bar b_{34} & |a - b|^2 + |b_{43}|^2 + |a_{12}|^2 
    	\end{bmatrix}, \\
    	{Z_{22}}^{(22)} = 
    	\begin{bmatrix}
    		|a - b|^2 + |b_{34}|^2 + |a_{21}|^2 & (a - b) \bar b_{43} + (\bar a - \bar b) b_{34} \\
    		(\bar a - \bar b) b_{43} + (a - b) \bar b_{34} & |a - b|^2 + |b_{43}|^2 + |a_{21}|^2
    	\end{bmatrix}.
    \end{align}
    The largest eigenvalue of a positive semidefinite matrix is bounded above by the sum of the largest eigenvalues of two diagonal blocks of the positive semidefinite matrix.
    \begin{align}
    	\sigma_1(Y_{11})^2 + \sigma_1(Y_{22})^2 &= \lambda_{\max}(Z_{11}) + \lambda_{\max}(Z_{22}) \\ 
    	&\leq \lambda_{\max}({Z_{11}}^{(11)}) + \lambda_{\max}({Z_{11}}^{(22)}) + \lambda_{\max}({Z_{22}}^{(11)}) + \lambda_{\max}({Z_{22}}^{(22)}). \label{ineq: 2nd case 1}
    \end{align}
    Without loss of generality, we can only work with $Z_{11}^{(11)}$ and apply the result to $Z_{11}^{(22)}$, $Z_{22}^{(11)}$ and $Z_{22}^{(22)}$. We first decompose $Z_{11}^{(11)}$ into the following form
    \begin{align}
    	Z_{11}^{(11)} = | a_{12} |^2 I_2 + \tilde Z_{11}^{(11)}.
    \end{align}
    Since 
    \begin{align}
    	\det \tilde Z_{11}^{(11)} &= (|b + a|^2 + |b_{12}|^2)(|b + a|^2 + |b_{21}|^2) - |(b + a) \bar b_{21} + (\bar b + \bar a) b_{12}|^2 \\
    	&\geq (|b + a|^2 + |b_{12}|^2)(|b + a|^2 + |b_{21}|^2) - (|(b + a) \bar b_{21}| + |(\bar b + \bar a) b_{12}|)^2 \label{ineq: 2nd case 6} \\
    	&= |b+a|^4 + |b_{12} b_{21}|^2 - 2|b + a|^2 |b_{21}||b_{12}| \geq 0, \label{ineq: 2nd case 7}
    \end{align}
    and $\operatorname{Tr} \tilde Z_{11}^{(11)} \geq 0$, $\tilde Z_{11}^{(11)}$ is positive semi-definite, we have
    \begin{align}
    	\lambda_{\max}(Z_{11}^{(11)}) \leq |a_{12}|^2 + 2 |b+a|^2 + |b_{12}|^2 + |b_{21}|^2. \label{ineq: 2nd case 2}
    \end{align}
    Similarly, we obtain
    \begin{align}
    	\lambda_{\max}(Z_{11}^{(22)}) &\leq |a_{21}|^2 + 2 |b+a|^2 + |b_{12}|^2 + |b_{21}|^2, \label{ineq: 2nd case 3}\\
    	\lambda_{\max}(Z_{22}^{(11)}) &\leq |a_{12}|^2 + 2 |b-a|^2 + |b_{34}|^2 + |b_{43}|^2, \label{ineq: 2nd case 4}\\
    	\lambda_{\max}(Z_{22}^{(22)}) &\leq |a_{21}|^2 + 2 |b-a|^2 + |b_{34}|^2 + |b_{43}|^2. \label{ineq: 2nd case 5}
    \end{align}
    Finally, we have
    \begin{align}
    	\sigma_1(Y_{11})^2 + \sigma_1(Y_{22})^2 &\leq \lambda_{\max}({Z_{11}}^{(11)}) + \lambda_{\max}({Z_{11}}^{(22)}) + \lambda_{\max}({Z_{22}}^{(11)}) + \lambda_{\max}({Z_{22}}^{(22)}) \\
    	&\leq 4|b+a|^2 + 4|b-a|^2 + 2(|b_{34}|^2 + |b_{43}|^2 + |b_{12}|^2 + |b_{21}|^2 + |a_{21}|^2 + |a_{21}|^2) \leq \frac{1}{2}.
    \end{align}
    
    The sum of the two largest singular values squared equals to $\frac{1}{2}$ if and only if \eqref{ineq: 2nd case 1}, \eqref{ineq: 2nd case 2}, \eqref{ineq: 2nd case 3}, \eqref{ineq: 2nd case 4} and \eqref{ineq: 2nd case 5} hold with equalities. \eqref{ineq: 2nd case 1} is rather difficult to tackle, we ignore it for now. Regarding \eqref{ineq: 2nd case 2}-\eqref{ineq: 2nd case 5}, it suffices to obtain under what conditions \eqref{ineq: 2nd case 2} holds with equality and apply the same method to \eqref{ineq: 2nd case 3}-\eqref{ineq: 2nd case 5}.
    
    Actually, Eq. \eqref{ineq: 2nd case 2} holds with equality if and only if \eqref{ineq: 2nd case 6} and \eqref{ineq: 2nd case 7} hold with equalities, which is equivalent to
    \begin{align}
    	&(b+a)\bar b_{21} \text{ and } (\bar b + \bar a)b_{12} \text{ are of the same phase,} \\
    	&|b+a|^2 - |b_{12}b_{21}| = 0.
    \end{align}
    Let $m_1 = |b+a|, m_2 = |b_{12}|, m_3 = |b_{21}|, e^{i\varphi_1} = \frac{b+a}{|b+a|}, e^{i\varphi_2} = \frac{b_{12}}{|b_{12}|}, e^{i\varphi_3} = \frac{b_{21}}{|b_{21}|}$. We have
    \begin{align}
    	m_1 m_3 &= m_1 m_2, \\
    	\varphi_1 - \varphi_3 &= \varphi_2 - \varphi_1, \\
    	m_1^2 &= m_2 m_3.
    \end{align}
    So $2\varphi_1 = \varphi_2 + \varphi_3$, $b+a = \sqrt{b_{12} b_{21}}$. We apply the same method to \eqref{ineq: 2nd case 3}-\eqref{ineq: 2nd case 5}, and obtain
    \begin{align}
    	b+a &= \sqrt{b_{12} b_{21}}, \label{eq: b+a}\\
    	a-b &= \sqrt{b_{34} b_{43}}. \label{eq: b-a}
    \end{align}
    These two equalities above are necessary for the sum of the two largest singular values squared to equal $\frac{1}{2}$. 
    
    We take advantage of these two equalities and prove the case \ref{2nd case: 2 2 bloc-diag} of the distillability conjecture in a new way. With this new method, we find out more necessary conditions of saturation.
    
    Let us first reformulate $\left\| A \right\|_F^2 + \left\| B \right\|_F^2 = \frac{1}{4}$. Note that $a_{34} = a_{43} = 0$ because $a_{34}$ and $a_{43}$ are not present in either $Y_{11}$ or $Y_{22}$. We have
    \begin{equation}
    	\begin{aligned}
    		\frac{1}{4} &= 4|a|^2 + 4|b|^2 + |a_{12}|^2 + |a_{21}|^2 + |b_{12}|^2 + |b_{21}|^2 + |b_{34}|^2 + |b_{43}|^2 \\
    		&= 2 (|a + b|^2 + |a - b|^2) + |a_{12}|^2 + |a_{21}|^2 + |b_{12}|^2 + |b_{21}|^2 + |b_{34}|^2 + |b_{43}|^2 \\
    		&= 2|b_{12} b_{21}| + 2|b_{34} b_{43}| + |a_{12}|^2 + |a_{21}|^2 + |b_{12}|^2 + |b_{21}|^2 + |b_{34}|^2 + |b_{43}|^2 \\
    		&= (|b_{12}| + |b_{21}|)^2 + (|b_{34}| + |b_{43}|)^2 + |a_{12}|^2 + |a_{21}|^2. \label{eq: norm cond}
    	\end{aligned}
    \end{equation}
    Then we partition $Y_{11}$ and $Y_{22}$ as follows,
    \begin{align}
    	Y_{11} = \begin{bmatrix}
    		P & a_{12} I \\
    		a_{21} I & P
    	\end{bmatrix}, \quad
    	Y_{22} = \begin{bmatrix}
    		Q & a_{12} I \\
    		a_{21} I & Q
    	\end{bmatrix},
    \end{align}
    with 
    \begin{align}
    	P = \begin{bmatrix}
    		\sqrt{b_{21} b_{12}} & b_{12} \\
    		b_{21} & \sqrt{b_{21} b_{12}}
    	\end{bmatrix}, \quad
    	Q = \begin{bmatrix}
    		\sqrt{b_{43} b_{34}} & b_{34} \\
    		b_{43} & \sqrt{b_{43} b_{34}}
    	\end{bmatrix}.
    \end{align}
    We notice that $P$ and $Q$'s determinants are zero, thus $P$ and $Q$ both have only one nonzero singular value. We perform singular value decompositions to $P$ and $Q$, 
    \begin{align}
    	P = T \Sigma_p U^*, &\quad Q = V \Sigma_q W^*, \label{eq: normalPQ}\\
    	\Sigma_p = \mathrm{diag}(p, 0), &\quad \Sigma_q = \mathrm{diag}(q, 0), \\
    	p = \left\| P \right\|_F = |b_{12}| + |b_{21}|, &\quad q = \left\| Q \right\|_F = |b_{34}| + |b_{43}|. \label{eq: b_{ij} and p,q}
    \end{align}
    There are three cases, (1) $p = q = 0$, (2) $p \neq 0, q \neq 0$, and (3) $p \text{ or } q = 0$.
    
    \begin{enumerate}
    	\item[(1)] We first work with the case where $p = q = 0$.
    \end{enumerate}
    
    This means $|b_{12}| = |b_{21}| = |b_{34}| = |b_{43}| = 0$. Recall that $b+a = \sqrt{b_{12}b_{21}}$ and $b-a = \sqrt{b_{34}b_{43}}$, so $b = a = 0$. Now $\sigma_1(Y_{11}) = \sigma_1(Y_{22}) = \max(|a_{12}|, |a_{21}|)$. Thus, $\sigma_1^2(Y_{11}) + \sigma_1^2(Y_{22}) = \frac{1}{2}$ if and only if $|a_{12}| = \frac{1}{2}, a_{21} = 0$ or $|a_{21}| = \frac{1}{2}, a_{12} = 0$. In fact, $|a_{12}| = \frac{1}{2}, a_{21} = 0$ and $|a_{21}| = \frac{1}{2}, a_{12} = 0$ are equivalent by taking the transpose of $A$. In this case, $A$ has only one nonzero entry and $B = 0$, which is already covered by the previous discussion
    \begin{align}
    	A = \begin{bmatrix}
    		0 & a_{12} & 0 & 0 \\
    		0 & 0 & 0 & 0 \\
    		0 & 0 & 0 & 0 \\
    		0 & 0 & 0 & 0
    	\end{bmatrix}, B = 0, |a_{12}| = \frac{1}{2}. 
    \end{align}
    
    \begin{enumerate}
    	\item[(2)] Then, we work with the case where $p$ and $q$ are both nonzero. 
    \end{enumerate}
    
    We consider $C_{11} = Y_{11}^* Y_{11}$. By calculations, we have
    \begin{align}
    	C_{11} &= Y_{11}^* Y_{11} = \begin{bmatrix}
    		P^* P + |a_{21}|^2 I & a_{12} P^* + \bar a_{21} P \\
    		\bar a_{12} P + a_{21} P^* & |a_{12}|^2 I + P^* P
    	\end{bmatrix} \label{eq: C11}\\
    	&= \begin{bmatrix}
    		P^* P & a_{12} P^* \\
    		\bar a_{12} P & |a_{12}|^2 I
    	\end{bmatrix}
    	+
    	\begin{bmatrix}
    		|a_{21}|^2 I & \bar a_{21} P \\
    		a_{21} P^* & P^* P
    	\end{bmatrix} \\
    	&= \underbrace{\begin{bmatrix}
    			U & 0 \\
    			0 & T
    		\end{bmatrix}
    		\begin{bmatrix}
    			\Sigma_p^2 & a_{12} \Sigma_p \\
    			\bar a_{12} \Sigma_p & |a_{12}|^2 I 
    		\end{bmatrix}
    		\begin{bmatrix}
    			U^* & 0 \\
    			0 & T^*
    	\end{bmatrix}}_{C_{11}^{(1)}}
    	+
    	\underbrace{\begin{bmatrix}
    			T & 0 \\
    			0 & U
    		\end{bmatrix}
    		\begin{bmatrix}
    			|a_{21}|^2 I & \bar a_{21} \Sigma_p \\
    			a_{21} \Sigma_p & \Sigma_p^2 
    		\end{bmatrix}
    		\begin{bmatrix}
    			T^* & 0 \\
    			0 & U^*
    	\end{bmatrix}}_{C_{11}^{(2)}}.
    \end{align}
    By permutation similarity, we have 
    \begin{align}
    	&\begin{bmatrix}
    		\Sigma_p^2 & a_{12} \Sigma_p \\
    		\bar a_{12} \Sigma_p & |a_{12}|^2 I 
    	\end{bmatrix} \text{ is equivalent to } 
    	\begin{bmatrix}
    		p^2 & a_{12} p & 0 & 0 \\
    		\bar a_{12} p & |a_{12}|^2 & 0 & 0 \\
    		0 & 0 & 0 & 0 \\
    		0 & 0 & 0 & |a_{12}|^2
    	\end{bmatrix}, \\
    	&\text{ whose largest eigenvalues are } |a_{12}|^2 + p^2,\\
    	&\begin{bmatrix}
    		|a_{21}|^2 I & \bar a_{21} \Sigma_p \\
    		a_{21} \Sigma_p & \Sigma_p^2 
    	\end{bmatrix} \text{ is equivalent to } 
    	\begin{bmatrix}
    		|a_{21}|^2 & 0 & 0 & 0 \\
    		0 & |a_{21}|^2 & \bar a_{21} p & 0 \\
    		0 & a_{21} p & p^2 & 0 \\
    		0 & 0 & 0 & 0
    	\end{bmatrix}, \\
    	&\text{ whose largest eigenvalues are } |a_{21}|^2 + p^2.
    \end{align}
    Thus
    \begin{align}
    	\sigma_1(Y_{11})^2 = \lambda_{\max}(C_{11}) \leq \lambda_{\max}(C_{11}^{(1)}) + \lambda_{\max}(C_{11}^{(2)}) = 2 p^2 + |a_{12}|^2 + |a_{21}|^2. \label{eq: ineq11}
    \end{align}
    Similarly, we apply the same method to $C_{22} = Y_{22}^* Y_{22}$, we have
    \begin{align}
    	C_{22} &= Y_{22}^* Y_{22} = C_{22}^{(1)} + C_{22}^{(2)} \\ 
    	&= \begin{bmatrix}
    		W & 0 \\
    		0 & V
    	\end{bmatrix}
    	\begin{bmatrix}
    		\Sigma_q^2 & a_{12} \Sigma_q \\
    		\bar a_{12} \Sigma_q & |a_{12}|^2 I 
    	\end{bmatrix}
    	\begin{bmatrix}
    		W^* & 0 \\
    		0 & V^*
    	\end{bmatrix}
    	+
    	\begin{bmatrix}
    		V & 0 \\
    		0 & W
    	\end{bmatrix}
    	\begin{bmatrix}
    		|a_{21}|^2 I & \bar a_{21} \Sigma_q \\
    		a_{21} \Sigma_q & \Sigma_q^2 
    	\end{bmatrix}
    	\begin{bmatrix}
    		V^* & 0 \\
    		0 & W^*
    	\end{bmatrix},
    \end{align}
    and
    \begin{align}
    	\sigma_1(Y_{22})^2 = \lambda_{\max}(C_{22}) \leq 2 q^2 + |a_{12}|^2 + |a_{21}|^2. \label{eq: ineq22}
    \end{align}
    According to \eqref{eq: norm cond}, we have $\sigma_1(Y_{11})^2 + \sigma_1(Y_{22})^2 \leq 2 p^2 + 2q^2 +2 |a_{12}|^2 + 2 |a_{21}|^2 = \frac{1}{2}$. 
    
    Now we have proven the distillability conjecture of case \ref{2nd case: 2 2 bloc-diag} in a new way, we can find out more necessary conditions of saturation.
    
    $\sigma_1(Y_{11})^2 + \sigma_1(Y_{22})^2 = \frac{1}{2}$, if and only if \eqref{eq: ineq11} and \eqref{eq: ineq22} take equalities, which is equivalent to
    \begin{align}
    	\exists x_1 \in \mathbb{C}^4, &C_{11} x_1 = (2p^2 + |a_{12}|^2 + |a_{21}|^2) x_1, \label{eq: Z11 eig} \\
    	&C_{11}^{(1)} x_1 = (p^2 + |a_{12}|^2) x_1, \label{eq: Z111 eig} \\
    	&C_{11}^{(2)} x_1 = (p^2 + |a_{21}|^2) x_1, \label{eq: Z112 eig} \\
    	\exists x_2 \in \mathbb{C}^4, &C_{22} x_2 = (2q^2 + |a_{12}|^2 + |a_{21}|^2) x_2, \\
    	&C_{22}^{(1)} x_2 = (q^2 + |a_{12}|^2) x_2, \\
    	&C_{22}^{(2)} x_2 = (q^2 + |a_{21}|^2) x_2. 
    \end{align}
    We first work with $C_{11}, C_{11}^{(1)}, C_{11}^{(2)}$ and $x_1$. By calculations, we obtain $x_1^{(1)}$ and $x_1^{(2)}$, the corresponding normalized eigenvectors of respectively $C_{11}^{(1)}$ and $C_{11}^{(2)}$'s largest eigenvalues
    \begin{align}
    	x_1^{(1)} = e^{i \beta} \frac{1}{\sqrt{p^2 + |a_{12}|^2}} \begin{bmatrix}
    		U & 0 \\
    		0 & T
    	\end{bmatrix} 
    	\begin{bmatrix}
    		p \\
    		0 \\
    		\bar a_{12} \\
    		0
    	\end{bmatrix}, \quad
    	x_1^{(2)} = e^{i \gamma}\frac{1}{\sqrt{p^2 + |a_{21}|^2}} \begin{bmatrix}
    		T & 0 \\
    		0 & U
    	\end{bmatrix} 
    	\begin{bmatrix}
    		\bar a_{21} \\
    		0 \\
    		p \\
    		0
    	\end{bmatrix},
    \end{align}
    where $\beta, \gamma \in \mathbb{R}$. Suppose that $x_1 = [\xi_1, \xi_2, \xi_3, \xi_4]^\top$, according to \eqref{eq: Z11 eig}, \eqref{eq: Z111 eig} and \eqref{eq: Z112 eig}, we have: 
    \begin{align}
    	x_1 = x_1^{(1)} = x_1^{(2)}.
    \end{align}
    Then
    \begin{align}
    	e^{i \gamma} \frac{1}{\sqrt{p^2 + |a_{21}|^2}} T \begin{bmatrix} \bar a_{21} \\ 0 \end{bmatrix} = \begin{bmatrix} \xi_1 \\ \xi_2 \end{bmatrix} = e^{i \beta} \frac{1}{\sqrt{p^2 + |a_{12}|^2}} U \begin{bmatrix} p \\ 0 \end{bmatrix} \label{eq: UT1} \\ 
    	= e^{i (\beta - \gamma)}\frac{\sqrt{p^2 + |a_{21}|^2}}{\sqrt{p^2 + |a_{12}|^2}} \begin{bmatrix} \xi_3 \\ \xi_4 \end{bmatrix} = e^{i (2 \beta - \gamma)}\frac{\sqrt{p^2 + |a_{21}|^2}}{p^2 + |a_{12}|^2} T \begin{bmatrix} \bar a_{12} \\ 0 \end{bmatrix} \label{eq: UT2}. 
    \end{align}
    Recall that $T$ and $U$ are unitary matrices with orthonormal columns, where we partition $U = [U_1 \ U_2]$ and $T = [T_1 \ T_2]$, $U_i, T_i \in \mathbb{C}^2$. Thus, we have:
    \begin{gather}
    	e^{2i\gamma} \frac{\bar a_{21}}{\sqrt{p^2 + |a_{21}|^2}} = e^{2i\beta} \frac{\bar a_{12} \sqrt{p^2 + |a_{21}|^2}}{p^2 + |a_{12}|^2}, \label{eq: A entry} \\
    	e^{i \gamma}\frac{\bar a_{21}}{\sqrt{p^2 + |a_{21}|^2}} T_1 = e^{i \beta}\frac{p}{\sqrt{p^2 + |a_{12}|^2}} U_1. \label{eq: unitary}
    \end{gather}
    According to \eqref{eq: A entry} and \eqref{eq: unitary}, we have $|a_{12}a_{21}| = p^2$. According to \eqref{eq: unitary}, we know that $T_1$ and $U_1$ are colinear, since $U_1, U_2$ and $T_1, T_2$ are two orthonormal basis on $\mathbb{C}^2$, $U_2$ and $T_2$ are also colinear, which means $U = T \mathrm{diag}(e^{i\psi}, e^{i\phi})$. Recall that $P = T \Sigma_p U^*$, thus, $P$ is a normal matrix, which is equivalent to $2 \sqrt{b_{21} b_{12}} = \lambda_{\max}(P) = \sigma_1(P) = |b_{12}| + |b_{21}|$. Consequently, we have $b_{21} = \bar{b}_{12}$. Then we apply the same method to $C_{22}, C_{22}^{(1)}, C_{22}^{(2)}$ and $x_2$, we obtain that $Q$ is normal, $b_{34} = \bar{b}_{43}$, and $|a_{12} a_{21}| = q^2$. 
    
    Now let us recapitulate the conditions we have obtained:
    \begin{align}
    	&b_{12} = \bar{b}_{21}, \label{eq: 69} \\
    	&b_{34} = \bar{b}_{43}, \\
    	&(|b_{34}| + |b_{43}|)^2 = q^2 = |a_{12} a_{21}| = p^2 = (|b_{12}| + |b_{21}|)^2, \label{eq: a12a21}\\
    	&a + b = \sqrt{b_{12} b_{21}}, \\
    	&a - b = \sqrt{b_{43} b_{34}}. 
    \end{align} 
    Consequently, we have:
    \begin{align}
    	&|b_{12}| = |b_{21}| = |b_{34}| = |b_{43}|, \\
        &a = |b_{12}|, \\
        &b = 0.
    \end{align}
    Without loss of generality, we can assume that $b_{12} = b_{34} = c$, then we have:
    \begin{align}
        A = \begin{bmatrix}
            |c| & a_{12} & 0 & 0 \\
            a_{21} & |c| & 0 & 0 \\
            0 & 0 & -|c| & 0 \\
            0 & 0 & 0 & -|c|
        \end{bmatrix}, \quad
        B = \begin{bmatrix}
            0 & c & 0 & 0 \\
            \bar c & 0 & 0 & 0 \\
            0 & 0 & 0 & c \\
            0 & 0 & \bar c & 0
        \end{bmatrix}.
    \end{align}
    We can see that $B$ is unitarily similar to $\operatorname{diag}(|c|, |c|, -|c|, -|c|)$, thus, we let $B = \operatorname{diag}(|c|, |c|, -|c|, -|c|)$. Now we calculate $X$:
    \begin{align}
        X = X_1 \oplus X_2,
    \end{align}
    where
    \begin{align}
        X_1 = \begin{bmatrix}
            2|c|I_2 \oplus 0 & a_{12} I_4 \\
            a_{21} I_4 & 2|c|I_2 \oplus 0
        \end{bmatrix}, \quad X_2 = \operatorname{diag}(0,0,-2|c|,-2|c|,0,0,-2|c|,-2|c|).
    \end{align}
    By calculations, $X_1$ is unitarily invariant to $X_1^{(11)} \oplus X_1^{(11)} \oplus X_1^{(22)}$, where
    \begin{align}
        X_1^{(11)} = \begin{bmatrix}
            2|c| & a_{12} \\
            a_{21} & 2|c|
        \end{bmatrix}, \quad 
        X_1^{(22)} = \begin{bmatrix}
            0 & a_{12} I_2 \\
            a_{21} I_2 & 0
        \end{bmatrix}.
    \end{align}
    If $\sigma_1(X) = \sigma_1(X_1^{(11)})$ and $\sigma_2(X) = \sigma_2(X_1^{(11)})$, $\sigma_1^2(X_1^{(11)}) + \sigma_2^2(X_1^{(11)}) = \|X_1^{(11)}\|_F^2 = \frac{1}{4} < \frac{1}{2}$, thus, this case is excluded. Besides, the case where $\sigma_1(X)$ and $\sigma_2(X)$ come from the diagonal part is also excluded, as it is included in previous discussion regarding normal $A$ and $B$. Now, we have $\sigma_1(X) = \sigma_1(X_1^{(11)})$ and $\sigma_2(X) = \sigma_1(X_1^{(11)})$. To maximize $\sigma_1^2(X) + \sigma_2^2(X)$, $X_1^{(11)}$ need to be of rank $1$, which means $b_{12} b_{21} = 4|c|^2$. This case is included in case \ref{1st case: 2 2 bloc-diag}.
    
    \begin{enumerate}[resume]
    	\item[(3)] Finally, we work with the case where $p$ or $q$ equals $0$.
    \end{enumerate}
    
    We first work with $p=0, q\neq0$ and apply the same method to $q=0, p\neq0$. Since $p=0$, according to \eqref{eq: norm cond}, \eqref{eq: b_{ij} and p,q} and \eqref{eq: b+a}, we have $(|b_{34}| + |b_{43}|)^2 + |a_{12}|^2 + |a_{21}|^2 = \frac{1}{4}$, $b_{12} = b_{21} = 0$ and $b+a = 0$. $Y_{11}$ and $Y_{22}$ are as follows:
    \begin{align}
    	Y_{11} &= 
    	\begin{bmatrix}
    		0 & 0 & a_{12} & 0 \\
    		0 & 0 & 0 & a_{21} \\
    		a_{12} & 0 & 0 & 0 \\
    		0 & a_{21} & 0 & 0
    	\end{bmatrix}, \\
    	Y_{22} &= \begin{bmatrix}
    		\sqrt{b_{43} b_{34}} & b_{34} & a_{12} & 0 \\
    		b_{43} & \sqrt{b_{43} b_{34}} & 0 & a_{12} \\
    		a_{21} & 0 & \sqrt{b_{43} b_{34}} & b_{34} \\
    		0 & a_{21} & b_{43} & \sqrt{b_{43} b_{34}}
    	\end{bmatrix}.
    \end{align}
    It is trivial that $\sigma_1(Y_{11}) = \max(|a_{12}|, |a_{21}|)$. According to \eqref{eq: ineq22}, we have $\sigma_1^2(Y_{22}) \leq 2q^2 + |a_{12}|^2 + |a_{21}|^2$. 
    \begin{align}
    	\sigma_1^2(Y_{11}) + \sigma_1^2(Y_{22}) \leq \max(|a_{12}|^2, |a_{21}|^2) + 2q^2 + |a_{12}|^2 + |a_{21}|^2 \leq \frac{1}{2}. \label{eq: p=0}
    \end{align}
    According to the discussion on $Y_{22}^* Y_{22}$ in case (2), the first inequality in \eqref{eq: p=0} takes equality only if $b_{34} = \bar b_{43}$ and $|a_{12} a_{21}| = q^2$. The second inequality in \eqref{eq: p=0} takes equality if and only if $a_{21} = 0$ or $a_{12} = 0$. According to \eqref{eq: b_{ij} and p,q}, $|b_{43}| = |b_{34}| = 0$. Now, let us recapitulate what we obtain
    \begin{gather}
    	b_{12} = b_{21} = b_{34} = b_{43} = a_{34} = a_{43} = 0, \\
    	a_{12} =0 \text{ or } a_{21} = 0, \\
    	b+a = \sqrt{b_{12}b_{21}} = b-a = \sqrt{b_{34}b_{43}} = 0. 
    \end{gather}
    This reduces to the case where $A$ only has one nonzero entry and $B = 0$. We can apply the same method to $q=0, p\neq0$ and get the same result. The case where $A$ or $B$ is zero is already covered by former discussion.
\end{proof}

\textbf{Lastly, we consider under what condition the inequality is saturated in case \ref{3rd case: 2 2 bloc-diag}}

\begin{lemma}\label{lemma: saturation of 3rd case: 2 2 bloc-diag}
    In case \ref{3rd case: 2 2 bloc-diag}, the saturation condition is included in case \ref{1st case: 2 2 bloc-diag}.
\end{lemma}

\begin{proof}
    We first review how to prove case \ref{3rd case: 2 2 bloc-diag}, and find out some necessary conditions of saturation based on the proof.

    $Y_{11}$ and $Y_{44}$ can be decomposed into the following form
    \begin{align}
    	&Y_{11} = {Y_{11}}^{(1)} + {Y_{11}}^{(2)} = 
    	\begin{bmatrix}
    		b+a & 0 & 0 & 0 \\
    		0 & b+a & 0 & 0 \\
    		0 & 0 & b+a & 0 \\
    		0 & 0 & 0 & b+a
    	\end{bmatrix}
    	+
    	\begin{bmatrix}
    		0 & b_{12} & a_{12} & 0 \\
    		b_{21} & 0 & 0 & a_{12} \\
    		a_{21} & 0 & 0 & b_{12} \\
    		0 & a_{21} & b_{21} & 0
    	\end{bmatrix}, \\
    	&Y_{44} = {Y_{44}}^{(1)} + {Y_{44}}^{(2)} = 
    	\begin{bmatrix}
    		-b-a & 0 & 0 & 0 \\
    		0 & -b-a & 0 & 0 \\
    		0 & 0 & -b-a & 0 \\
    		0 & 0 & 0 & -b-a
    	\end{bmatrix}
    	+
    	\begin{bmatrix}
    		0 & b_{34} & a_{34} & 0 \\
    		b_{43} & 0 & 0 & a_{34} \\
    		a_{43} & 0 & 0 & b_{34} \\
    		0 & a_{43} & b_{43} & 0
    	\end{bmatrix}.
    \end{align}
    Then, according to the triangular inequality of the operator norm, which is the largest singular value in our scenario, we have
    \begin{align}
    	\sigma_1(Y)^2 + \sigma_2(Y)^2 &= \sigma_1(Y_{11})^2 + \sigma_1(Y_{44})^2 \leq \left(\sigma_1({Y_{11}}^{(1)}) + \sigma_1({Y_{11}}^{(2)})\right)^2 + \left(\sigma_1({Y_{44}}^{(1)}) + \sigma_1({Y_{44}}^{(2)})\right)^2 \label{ineq: 104}\\
    	&\leq 2 \left(\sigma_1({Y_{11}}^{(1)})^2 + \sigma_1({Y_{11}}^{(2)})^2 + \sigma_1({Y_{44}}^{(1)})^2 + \sigma_1({Y_{44}}^{(2)})^2\right) \label{ineq: 105}.
    \end{align}
    Without loss of generality, we can only work with $Y_{11}$ and apply the result to $Y_{44}$. In fact, there are permutation matrices $P$ and $Q$ such that
    \begin{align}
    	\tilde Y_{11}^{(2)} = P Y_{11}^{(2)} Q = \begin{bmatrix}
    		a_{21} & b_{12} & 0 & 0 \\
    		b_{21} & a_{12} & 0 & 0 \\
    		0 & 0 & a_{12} & b_{12} \\
    		0 & 0 & b_{21} & a_{21}
    	\end{bmatrix}, \label{eq: 106}
    \end{align}
    Then, we obtain 
    \begin{align}
    	\sigma_1(Y_{11}^{(2)})^2 \leq |a_{21}|^2 + |a_{12}|^2 + |b_{12}|^2 + |b_{21}|^2. \label{ineq: 107}
    \end{align}
    In the same way, we obtain
    \begin{align}
    	\sigma_1({Y_{44}}^{(2)})^2 \leq |a_{43}|^2 + |a_{34}|^2 + |b_{34}|^2 + |b_{43}|^2. \label{ineq: 108}
    \end{align}
    Now, we have
    \begin{align}
    	\sigma_1(Y)^2 + \sigma_2(Y)^2 &\leq 2 \left(\sigma_1({Y_{11}}^{(1)})^2 + \sigma_1({Y_{11}}^{(2)})^2 + \sigma_1({Y_{44}}^{(1)})^2 + \sigma_1({Y_{44}}^{(2)})^2\right) \\
    	&= 2 \left( 2|b+a|^2 + |a_{43}|^2 + |a_{34}|^2 + |b_{34}|^2 + |b_{43}|^2 + |a_{21}|^2 + |a_{12}|^2 + |b_{12}|^2 + |b_{21}|^2 \right) \\
    	&= 2 \left( 4 |a|^2 + 4 |a|^2 + |a_{43}|^2 + |a_{34}|^2 + |b_{34}|^2 + |b_{43}|^2 + |a_{21}|^2 + |a_{12}|^2 + |b_{12}|^2 + |b_{21}|^2 \right).
    \end{align}
    Recall that $\left\| A \right\|_F^2 + \left\| B \right\|_F^2 = \frac{1}{4} = 4 |a|^2 + 4 |a|^2 + |a_{43}|^2 + |a_{34}|^2 + |b_{34}|^2 + |b_{43}|^2 + |a_{21}|^2 + |a_{12}|^2 + |b_{12}|^2 + |b_{21}|^2$, we have $\sigma_1(Y)^2 + \sigma_2(Y)^2 \leq \frac{1}{2}$.
    
    Now, we have proven case \ref{3rd case: 2 2 bloc-diag}, we will find out some necessary conditions of saturation. Since $b-a$ does not present in either $Y_{11}$ or $Y_{44}$, In order to maximize $\sigma_1^2(Y_{11}) + \sigma_1^2(Y_{22})$, $b - a$ needs to be zero. $\sigma_1^2(Y_{11}) + \sigma_1^2(Y_{22}) = \frac{1}{2}$ if and only if \eqref{ineq: 104}, \eqref{ineq: 105}, \eqref{ineq: 107} and \eqref{ineq: 108} take equalities. We first deal with the part involves $Y_{11}$
    \begin{align}
    	\sigma_1(Y_{11}) &\leq \sigma_1({Y_{11}}^{(1)}) + \sigma_1({Y_{11}}^{(2)}), \label{ineq: 112}\\
    	\left(\sigma_1({Y_{11}}^{(1)}) + \sigma_1({Y_{11}}^{(2)})\right)^2 &\leq 2\left( \sigma_1^2({Y_{11}}^{(1)}) + \sigma_1^2({Y_{11}}^{(2)}) \right), \label{ineq: 113}\\
    	\sigma_1(Y_{11}^{(2)})^2 &\leq |a_{21}|^2 + |a_{12}|^2 + |b_{12}|^2 + |b_{21}|^2. \label{ineq: 114}
    \end{align}
    Then, we apply the same method to the part involves $Y_{44}$
    \begin{align}
    	\sigma_1(Y_{44}) &\leq \sigma_1({Y_{44}}^{(1)}) + \sigma_1({Y_{44}}^{(2)}), \\
    	\left(\sigma_1({Y_{44}}^{(1)}) + \sigma_1({Y_{44}}^{(2)})\right)^2 &\leq 2\left( \sigma_1^2({Y_{44}}^{(1)}) + \sigma_1^2({Y_{44}}^{(2)}) \right), \\
    	\sigma_1(Y_{44}^{(2)})^2 &\leq |a_{43}|^2 + |a_{34}|^2 + |b_{34}|^2 + |b_{43}|^2. 
    \end{align}
    \eqref{ineq: 112} takes equality if and only if 
    \begin{align}
    	&Y_{11}^{(2)} y = \lambda_1 y, \label{eq: 118}\\
    	&|\lambda_1| = \sigma_1(Y_{11}^{(2)}), 
    \end{align}
    and $\lambda_1$ and $b+a$ are of the same phase. \eqref{ineq: 113} takes equality if and only if
    \begin{align}
    	\sigma_1(Y_{11}^{(1)}) = \sigma_1(Y_{11}^{(2)}) = |b+a|.
    \end{align}
    \eqref{ineq: 114} takes equality if and only if the two diagonal blocks of $\tilde Y_{11}^{(2)}$ are singular. According to what we have obtained above, we have:
    \begin{align}
    	\sigma_1(Y_{11}^{(1)}) = \sigma_1(Y_{11}^{(2)}) = |\lambda_1| = |b+a|, \\
    	\lambda_1 = b+a, b_{12}b_{21} = a_{12}a_{21}.
    \end{align}
    Under the assumption that $b_{12}b_{21} = a_{12}a_{21}$, the spectrum of $Y_{11}^{(2)}$ is $\{ 2\sqrt{a_{12}a_{21}}, -2\sqrt{a_{12}a_{21}}, 0, 0 \}$, thus, we have $b+a = 2\sqrt{a_{12}a_{21}}$ or $b+a = -2\sqrt{a_{12}a_{21}}$. Since $|\lambda_1| = \sigma_1({Y_{11}^{(2)}})$, we have
    \begin{align}
    	4|a_{12}a_{21}| = |a_{21}|^2 + |a_{12}|^2 + |b_{12}|^2 + |b_{21}|^2,
    \end{align}
    recall that $b_{12}b_{21} = a_{12}a_{21}$, we obtain
    \begin{align}
    	(|a_{21}| - |a_{12}|)^2 + (b_{21}| - |b_{12}|)^2 = 0,
    \end{align}
    thus, $|a_{12}| = |a_{21}|, |b_{12}| = |b_{21}|$. Since $b_{12}b_{21} = a_{12}a_{21}$, we have $|a_{12}| = |a_{21}| = |b_{12}| = |b_{21}|$. We apply the same method to $Y_{44}$, and recapitulate what we obtain
    \begin{gather}
    	a_{34}a_{43} = b_{34}b_{43}, b_{12}b_{21} = a_{12}a_{21}, \label{eq: 127}\\
    	b+a = 2\sqrt{a_{12}a_{21}}\text{ or }b+a = -2\sqrt{a_{12}a_{21}},  \label{eq: 128}\\
    	-b-a = 2\sqrt{a_{34}a_{43}}\text{ or }-b-a = -2\sqrt{a_{34}a_{43}},  \label{eq: 129} \\
    	|a_{12}| = |a_{21}| = |b_{12}| = |b_{21}|, \\
    	|a_{34}| = |a_{43}| = |b_{34}| = |b_{43}|, \\
    	b = a.
    \end{gather}
    From the last five we can derive that 
    \begin{align}
    	|a_{12}| = |a_{21}| = |a_{34}| = |a_{43}| = |b_{12}| = |b_{21}| = |b_{34}| = |b_{43}| = |a| = |b|. \label{eq: 132}
    \end{align}
    
    Now, we will prove that the conditions we obtain are sufficient.
    
    \eqref{eq: 127} implies that \eqref{ineq: 107} and \eqref{ineq: 108} take equalities. \eqref{eq: 132} implies that $|\lambda_1(Y_{11}^{(2)})| =  \sigma_1(Y_{11}^{(2)})$ and $|\lambda_1(Y_{44}^{(2)})| = \sigma_1(Y_{44}^{(2)})$, where $\lambda_1$ denotes the eigenvalue with largest absolute value. Plus, \eqref{eq: 128} and \eqref{eq: 129} imply that the eigenvalue of $Y_{11}^{(1)}$ (respectively $Y_{44}^{(1)}$) and the eigenvalue of $Y_{11}^{(2)}$ (respectively $Y_{44}^{(2)}$) are equal, Thus \eqref{ineq: 104} and \eqref{ineq: 105} takes equality. \eqref{ineq: 104}, \eqref{ineq: 105}, \eqref{ineq: 107} and \eqref{ineq: 108} all take equalities, $\sigma_1^2(Y_{11}) + \sigma_2^2(Y_{22}) = \frac{1}{2}$. In this case
    \begin{gather*}
    	A = \begin{bmatrix}
    		a & a_{12} & 0 & 0 \\
    		a_{21} & a & 0 & 0 \\
    		0 & 0 & -a & a_{34} \\
    		0 & 0 & a_{43} & -a
    	\end{bmatrix}, 
    	B = \begin{bmatrix}
    		a & b_{12} & 0 & 0 \\
    		b_{21} & a & 0 & 0 \\
    		0 & 0 & -a & b_{34} \\
    		0 & 0 & b_{43} & -a
    	\end{bmatrix}, \\
    	a_{34}a_{43} = b_{34}b_{43}, b_{12}b_{21} = a_{12}a_{21}, \\
    	b+a = 2\sqrt{a_{12}a_{21}}\text{ or }b+a = -2\sqrt{a_{12}a_{21}}, \\
    	-b-a = 2\sqrt{a_{34}a_{43}}\text{ or }-b-a = -2\sqrt{a_{34}a_{43}}, \\
    	|a_{12}| = |a_{21}| = |b_{12}| = |b_{21}| = |a_{34}| = |a_{43}| = |b_{34}| = |b_{43}| = |a|.
    \end{gather*}
    We notice that the diagonal blocks of $A$ and $B$ are all of rank $1$, since Trace and Frobenius norm are invariant under unitary similarity, we obtain that, up to unitary similarity
    \begin{align}
    	A = B = \begin{bmatrix}
    		2 a & 0 & 0 & 0 \\
    		0 & 0 & 0 & 0 \\
    		0 & 0 & -2 a & 0 \\
    		0 & 0 & 0 & 0
    	\end{bmatrix}.
    \end{align}
    It is included in case \ref{1st case: 2 2 bloc-diag}.
\end{proof}

\section{Proof of Lemma \ref{lemma:2 2 anti}}\label{appendix:proof 2 2 anti}

We first review how to prove the inequality $\sigma_1^2(X) + \sigma_2^2(X) \leq \frac{1}{2}$ \cite{SIO2025152}. Calculating $X^* X$, we find that $X^* X$ is permutationally similar to $H \oplus K$:
\begin{align}
    H=&\begin{bmatrix}
        I_2 \otimes B_2^* B_2+A_2^* A_2 \otimes I_2 & A_1 \otimes B_2^*+A_2^* \otimes B_1 \\
        A_1^* \otimes B_2+A_2 \otimes B_1^* & A_1^* A_1 \otimes I_2+I_2 \otimes B_1^* B_1
    \end{bmatrix},\\
    K=&\begin{bmatrix}
        I_2 \otimes B_1^* B_1+A_2^* A_2 \otimes I_2 & A_1 \otimes B_1^*+A_2^* \otimes B_2 \\
        A_1^* \otimes B_1+A_2 \otimes B_2^* & A_1^* A_1 \otimes I_2+I_2 \otimes B_2^* B_2
    \end{bmatrix}.
\end{align}
Let $\alpha_{ij} = \sigma_i(A_j), \beta_{ij} = \sigma_i(B_j)$. The matrix $H$ can be written as $H=H_1+H_2$, where
\begin{align}
    H_1=\begin{bmatrix}
        I_2 \otimes B_2^* B_2 & A_1 \otimes B_2^* \\
        A_1^* \otimes B_2 & A_1^* A_1 \otimes I_2
    \end{bmatrix}, \quad 
    H_2=\begin{bmatrix}
        A_2^* A_2 \otimes I_2 & A_2^* \otimes B_1 \\
        A_2 \otimes B_1^* & I_2 \otimes B_1^* B_1
    \end{bmatrix}.
\end{align} 
We have
\begin{align}
    \lambda(H_1)=\left\{\beta_{12}^2+\alpha_{11}^2, \beta_{12}^2+\alpha_{21}^2, \beta_{22}^2+\alpha_{11}^2, \beta_{22}^2+\alpha_{21}^2, 0,0,0,0\right\},
\end{align}
because
\begin{align}
    \lambda\left(H_1\right)=&\lambda\left(\begin{bmatrix}
    I_2 \otimes B_2^* \\
    A_1^* \otimes I_2
    \end{bmatrix}\begin{bmatrix}
    I_2 \otimes B_2 & A_1 \otimes I_2
    \end{bmatrix}\right) \\
    =& \lambda\left(\begin{bmatrix}
    I_2 \otimes B_2 & A_1 \otimes I_2
    \end{bmatrix}\begin{bmatrix}
    I_2 \otimes B_2^* \\
    A_1^* \otimes I_2
    \end{bmatrix}\right) \cup \{0, 0, 0, 0\}\\
    =& \lambda\left(A_1 A_1^* \otimes I + I \otimes B_2 B_2^*\right) \cup \{0, 0, 0, 0\} \\
    =& \{\beta_{12}^2 + \alpha_{11}^2, \beta_{12}^2 + \alpha_{21}^2, \beta_{22}^2 + \alpha_{11}^2, \beta_{22}^2 + \alpha_{21}^2, 0, 0, 0, 0\}.
\end{align}
Similarly, we set $K = K_1 + K_2$, where
\begin{align}
    K_1=\begin{bmatrix}
        I_2 \otimes B_1^* B_1 & A_1 \otimes B_1^* \\
        A_1^* \otimes B_1 & A_1^* A_1 \otimes I_2
    \end{bmatrix}, \quad 
    K_2=\begin{bmatrix}
        A_2^* A_2 \otimes I_2 & A_2^* \otimes B_2 \\
        A_2 \otimes B_2^* & I_2 \otimes B_2^* B_2
    \end{bmatrix}.
\end{align}
Then, we have
\begin{align}
    \lambda\left(H_2\right) = \{\beta_{11}^2 + \alpha_{12}^2, \beta_{11}^2 + \alpha_{22}^2, \beta_{21}^2 + \alpha_{12}^2, \beta_{21}^2 + \alpha_{22}^2, 0, 0, 0, 0\}, \\
    \lambda\left(K_1\right) = \{\beta_{11}^2 + \alpha_{11}^2, \beta_{11}^2 + \alpha_{21}^2, \beta_{21}^2 + \alpha_{11}^2, \beta_{21}^2 + \alpha_{21}^2, 0, 0, 0, 0\}, \\
    \lambda\left(K_2\right) = \{\beta_{12}^2 + \alpha_{12}^2, \beta_{12}^2 + \alpha_{22}^2, \beta_{22}^2 + \alpha_{12}^2, \beta_{22}^2 + \alpha_{22}^2, 0, 0, 0, 0\}.
\end{align}
There are three cases,
\begin{enumerate}
    \item $\sigma_1^2(X) = \lambda_1(H)$ and $\sigma_2^2(X) = \lambda_2(H)$. \label{case: 1st case 2 2 anti}
    \item $\sigma_1^2(X) = \lambda_1(K)$ and $\sigma_2^2(X) = \lambda_2(K)$. \label{case: 2nd case 2 2 anti}
    \item $\sigma_1^2(X) = \lambda_1(H)$ and $\sigma_2^2(X) = \lambda_1(K)$. \label{case: 3rd case 2 2 anti}
\end{enumerate}
\textbf{We first deal with case \ref{case: 1st case 2 2 anti}}.
\begin{align}
    \lambda_1(H) + \lambda_2(H) &\leq \lambda_1(H_1) + \lambda_2(H_1) + \lambda_1(H_2) + \lambda_2(H_2) \label{eq: 2 2 anti ineq 1}\\        &=\left(\beta_{12}^2+\alpha_{11}^2\right)+\left(\beta_{11}^2+\alpha_{12}^2\right)+\max \left\{\beta_{12}^2+\alpha_{21}^2, \beta_{22}^2+\alpha_{11}^2\right\}+\max \left\{\beta_{11}^2+\alpha_{22}^2, \beta_{21}^2+\alpha_{12}^2\right\} \\  &\leq\left(\beta_{12}^2+\alpha_{11}^2\right)+\left(\beta_{11}^2+\alpha_{12}^2\right)+\left(\beta_{12}^2+\alpha_{21}^2+\beta_{22}^2+\alpha_{11}^2\right)+\left(\beta_{11}^2+\alpha_{22}^2+\beta_{21}^2+\alpha_{12}^2\right) \label{eq: 2 2 anti ineq 2}\\
    &=\left(2 \alpha_{11}^2+2 \alpha_{12}^2+\alpha_{21}^2+\alpha_{22}^2\right)+\left(2 \beta_{11}^2+2 \beta_{12}^2+\beta_{21}^2+\beta_{22}^2\right) \\
    &\leq 2\|A\|_F^2+2\|B\|_F^2 \label{eq: 2 2 anti ineq 3}\\
    &=\frac{1}{2} .
\end{align}
Inequality \eqref{eq: 2 2 anti ineq 3} takes equality if and only if $\alpha_{21} = \alpha_{22} = \beta_{21} = \beta_{22} = 0$, which indicates that $\rank A_1 = \rank A_2 = \rank B_1 = \rank B_2 = 1$. Inequality \eqref{eq: 2 2 anti ineq 2} takes equality if and only if $\alpha_{11} \beta_{12}= \beta_{11} \alpha_{12} = 0$. We exclude the case where $\alpha_{11} = \alpha_{12} = 0$ or $\beta_{11} = \beta_{12} = 0$, which is already covered by Lemma \ref{lemma:B=0} and Lemma \ref{lemma:A=0}. Thus, we only need to deal with the case where $\alpha_{11} = \beta_{11} = 0$ or $\alpha_{12} = \beta_{12} = 0$, which indicates that $A_1 = B_1 = 0$ or $A_2 = B_2 = 0$. Consider only the case where $A_2 = B_2 = 0$, since the other one is similar. Now, we have 
\begin{align}
    A = \begin{bmatrix}
        0 & 0 & a_{13} & 0 \\
        0 & 0 & a_{23} & 0 \\
        0 & 0 & 0 & 0 \\
        0 & 0 & 0 & 0
    \end{bmatrix}, \quad
    B = \begin{bmatrix}
        0 & 0 & b_{13} & 0 \\
        0 & 0 & b_{23} & 0 \\
        0 & 0 & 0 & 0 \\
        0 & 0 & 0 & 0
    \end{bmatrix}.
\end{align}
Now, we calculate $H$
\begin{align}
    H = \begin{bmatrix}
        0 & 0 \\
        0 & A_1^* A_1 \otimes I_2 + I_2 \otimes B_1^* B_1
    \end{bmatrix},
\end{align}
where $A_1^* A_1 \otimes I_2 + I_2 \otimes B_1^* B_1 = \operatorname{diag}\left(|a_{13}|^2 + |a_{23}|^2 + |b_{13}|^2 + |b_{23}|^2, |a_{13}|^2 + |a_{23}|^2, |b_{13}|^2 + |b_{23}|^2, 0\right)$. Thus, $\lambda_1(H) + \lambda_2(H) = |a_{13}|^2 + |a_{23}|^2 + |b_{13}|^2 + |b_{23}|^2 + \max\left(|a_{13}|^2 + |a_{23}|^2, |b_{13}|^2 + |b_{23}|^2\right) = \frac{1}{2}$ if and only if $a_{13} = a_{23} = 0$ or $b_{13} = b_{23} = 0$, which indicates that $A = 0$ or $B = 0$, this case is already covered by Lemma \ref{lemma:B=0} and Lemma \ref{lemma:A=0}.

\textbf{We can obtain the same conclusion for Case \ref{case: 2nd case 2 2 anti}. Now, we deal with Case \ref{case: 3rd case 2 2 anti}.}
\begin{align}
    \lambda_1\left(X^* X\right)+\lambda_2\left(X^* X\right) & \leq \lambda_1\left(H_1\right)+\lambda_1\left(H_2\right)+\lambda_1\left(K_1\right)+\lambda_1\left(K_2\right) \label{eq: 2 2 anti ineq 7} \\
    & =\left(\beta_{12}^2+\alpha_{11}^2\right)+\left(\beta_{11}^2+\alpha_{12}^2\right)+\left(\beta_{11}^2+\alpha_{11}^2\right)+\left(\beta_{12}^2+\alpha_{12}^2\right) \\
    & =2\left(\beta_{11}^2+\beta_{12}^2+\alpha_{11}^2+\alpha_{12}^2\right) \\
    & \leq 2\|A\|_F^2+2\|B\|_F^2 \label{eq: 2 2 anti ineq 4}\\
    & =\frac{1}{2}.
\end{align}
Inequality \eqref{eq: 2 2 anti ineq 4} takes equality if and only if $\beta_{21} = \beta_{22} = \alpha_{21} = \alpha_{22} = 0$, which indicates that $\operatorname{rank} B_1 = \operatorname{rank} B_2 = \operatorname{rank} A_1 = \operatorname{rank} A_2 = 1$. Thus, by unitary similarity, we have
\begin{align}
    A = \begin{bmatrix}
        0 & 0 & a_{13} & 0 \\
        0 & 0 & a_{23} & 0 \\
        a_{31} & 0 & 0 & 0 \\
        a_{41} & 0 & 0 & 0
    \end{bmatrix}, \quad
    B = \begin{bmatrix}
        0 & 0 & b_{13} & 0 \\
        0 & 0 & b_{23} & 0 \\
        b_{31} & 0 & 0 & 0 \\
        b_{41} & 0 & 0 & 0
    \end{bmatrix}.
\end{align}
Now we deal with inequality \eqref{eq: 2 2 anti ineq 7}. Let the standard basis vector be denoted as follows.
\begin{align}
    e_1 &= \begin{bmatrix} 
        1 \\ 
        0 
    \end{bmatrix}.
\end{align}
We can define the vectors for the non-zero columns.
\begin{align}
    a_1 = \begin{bmatrix}
        a_{13} \\ 
        a_{23}
    \end{bmatrix}, \quad 
    b_1 = \begin{bmatrix} 
        b_{13} \\ 
        b_{23}
    \end{bmatrix}, \quad 
    a_2 = \begin{bmatrix} 
        a_{31} \\ 
        a_{41} 
    \end{bmatrix}, \quad 
    b_2 = \begin{bmatrix} 
        b_{31} \\ 
        b_{41}
    \end{bmatrix}.
\end{align}
This allows us to express the matrices as follows:
\begin{align}
    A_i &= a_i e_1^\top, \quad B_i = b_i e_1^\top.
\end{align}
We can factorize $H_1$ and $H_2$ into the following form.
\begin{align}
    H_1 &= V_1^* V_1, \quad V_1 = \begin{bmatrix} I_2 \otimes B_2 & A_1 \otimes I_2 \end{bmatrix}, \\
    H_2 &= V_2^* V_2, \quad V_2 = \begin{bmatrix} A_2 \otimes I_2 & I_2 \otimes B_1 \end{bmatrix}.
\end{align}
Since $V_i^* V_i$ and $V_i V_i^*$ share the same nonzero eigenvalues.
\begin{equation}
\begin{aligned}
    V_1 V_1^* &= I_2 \otimes (B_2 B_2^*) + (A_1 A_1^*) \otimes I_2 \\
    &= I_2 \otimes (b_2 b_2^*) + (a_1 a_1^*) \otimes I_2,
\end{aligned}
\end{equation}
Because $B_i$ and $A_i$ are rank-1 matrices, we have:
\begin{align}
    \lambda_1(H_1) &= \|a_1\|_F^2 + \|b_2\|_F^2.
\end{align}
The corresponding maximum eigenvector is $a_1 \otimes b_2$. Multiplying by $V_1^*$ yields the maximum eigenvector for $H_1$.
\begin{align}
    \Psi_1 &= V_1^* (a_1 \otimes b_2) = \begin{bmatrix} \|b_2\|_F^2 a_1 \otimes e_1 \\ \|a_1\|_F^2 e_1 \otimes b_2 \end{bmatrix}.
\end{align}
Applying the identical procedure to $H_2$ yields its maximum eigenvalue and corresponding eigenvector:
\begin{align}
    \lambda_1(H_2) &= \|a_2\|_F^2 + \|b_1\|_F^2, \\
    \Psi_2 &= \begin{bmatrix} 
        \|a_2\|_F^2 e_1 \otimes b_1 \\ 
        \|b_1\|_F^2 a_2 \otimes e_1 
    \end{bmatrix}.
\end{align}
By Weyl's inequality, the sum of the maximum eigenvalues bounds the maximum eigenvalue of the sum. For the equality condition to hold, the matrices must share the same maximum eigenvector. Therefore, there must exist a non-zero constant $c$ such that the eigenvectors are proportional.
\begin{align}
    \|b_2\|^2 a_1 \otimes e_1 &= c \|a_2\|^2 e_1 \otimes b_1, \label{eq: 2 2 anti ineq 5}\\
    \|a_1\|^2 e_1 \otimes b_2 &= c \|b_1\|^2 a_2 \otimes e_1. \label{eq: 2 2 anti ineq 6}
\end{align}
We now expand the Kronecker products in \eqref{eq: 2 2 anti ineq 5}:
\begin{align}
    a_1 \otimes e_1 &= \begin{bmatrix} 
        a_{13} \\ 
        0 \\ 
        a_{23} \\ 
        0 
    \end{bmatrix}, \quad 
    e_1 \otimes b_1 = \begin{bmatrix} 
        b_{13} \\ 
        b_{23} \\ 
        0 \\ 
        0 
    \end{bmatrix}.
\end{align}
For these vectors to be proportional, we have $a_{23} = b_{23} = 0$. We expand the Kronecker products from \eqref{eq: 2 2 anti ineq 6} in the same manner, and obtain $a_{41} = b_{41} = 0$. Thus, we have 
\begin{align}
    A = \begin{bmatrix}
        0 & 0 & a_{13} & 0 \\
        0 & 0 & 0 & 0 \\
        a_{31} & 0 & 0 & 0 \\
        0 & 0 & 0 & 0
    \end{bmatrix}, \quad
    B = \begin{bmatrix}
        0 & 0 & b_{13} & 0 \\
        0 & 0 & 0 & 0 \\
        b_{31} & 0 & 0 & 0 \\
        0 & 0 & 0 & 0
    \end{bmatrix}.
\end{align}
Applying permutation similarity to $A$ and $B$ with $P = \begin{bmatrix}
    1 & 0 & 0 & 0 \\
    0 & 0 & 1 & 0 \\
    0 & 1 & 0 & 0 \\
    0 & 0 & 0 & 1
\end{bmatrix}$, $A$ and $B$ become $2 \times 2$ block-diagonal matrices.

\section{Code used in Section \ref{sec:num evidence}}\label{appendix: numerical code}
\lstset{                 
	breaklines=true,           
    basicstyle=\small
}
\begin{lstlisting}
import numpy as np
import torch
import pymanopt
from pymanopt.manifolds import Sphere
from pymanopt.optimizers import ConjugateGradient

def generate_optimized_matrices(n: int):
	manifold = Sphere(24)
	@pymanopt.function.pytorch(manifold)
	def cost(v):
		v_scaled = v / 2.0
		A = torch.zeros(4, 4, dtype=v.dtype, device=v.device)
		B = torch.zeros(4, 4, dtype=v.dtype, device=v.device)
		mask = ~torch.eye(4, dtype=torch.bool, device=v.device)
		A[mask] = v_scaled[:12]
		B[mask] = v_scaled[12:]
		I = torch.eye(4, dtype=v.dtype, device=v.device)
		X = torch.kron(A, I) + torch.kron(I, B) 
		S = torch.linalg.svdvals(X)
		y = S[0]**2 + S[1]**2
	return -y
	
	problem = pymanopt.Problem(manifold=manifold, cost=cost)
	optimizer = ConjugateGradient(verbosity=0) 
	results = []
	total_attempts = 0
	while len(results) < n:
		total_attempts += 1
		v_init = np.random.randn(24)
		v_init = v_init / np.linalg.norm(v_init)
		result = optimizer.run(problem, initial_point=v_init)
		v_opt = result.point
		y_opt = -result.cost
		if y_opt >= 0.499999:
			v_scaled = v_opt / 2.0
			A_final = np.zeros((4, 4))
			B_final = np.zeros((4, 4))
			mask_np = ~np.eye(4, dtype=bool)
			A_final[mask_np] = v_scaled[:12]
			B_final[mask_np] = v_scaled[12:]
			results.append({"A": A_final, "B": B_final, "y": y_opt})
			print(f"Successfully found group {len(results)}/{n}! (Valid search iteration attempts this time: {total_attempts}, current y = {y_opt:.6f})")
			total_attempts = 0
	return results

if __name__ == "__main__":
	optimized_matrices = generate_optimized_matrices(n=10)
	for i, res in enumerate(optimized_matrices):
		print(f"\n--- Group {i+1} result (y={res['y']:.6f}) ---")
		print("Singular value of X: \n", np.round(np.linalg.svd(np.kron(res['A'], np.eye(4)) + np.kron(np.eye(4), res['B']), compute_uv=False), 4))
		print("The difference between the vector of singular values of A (sorted in descending order) and the vector of the absolute values of its eigenvalues (sorted in descending order): \n", np.round(np.linalg.svd(res['A'], compute_uv=False) - np.sort(np.abs(np.linalg.eigvals(res['A'])))[::-1], 4))
		print("The difference between the vector of singular values of B (sorted in descending order) and the vector of the absolute values of its eigenvalues (sorted in descending order): \n", np.round(np.linalg.svd(res['B'], compute_uv=False) - np.sort(np.abs(np.linalg.eigvals(res['B'])))[::-1], 4))
\end{lstlisting}

\bibliographystyle{unsrt}

\bibliography{citation}

\end{document}